\documentclass{article}

    \PassOptionsToPackage{numbers, compress}{natbib}
\usepackage[final]{neurips_2023}

\usepackage[utf8]{inputenc} % allow utf-8 input
\usepackage[T1]{fontenc}    % use 8-bit T1 fonts
\usepackage{url}            % simple URL typesetting
\usepackage{booktabs}       % professional-quality tables
\usepackage{amsfonts}       % blackboard math symbols
\usepackage{nicefrac}       % compact symbols for 1/2, etc.
\usepackage{microtype}      % microtypography
\usepackage[dvipsnames]{xcolor}         % colors
\usepackage{graphicx}
\usepackage{algorithm}
\usepackage{algpseudocode}
\usepackage{appendix}
\usepackage{multirow}
\usepackage{amsmath,amsfonts,bm}

\def\eqref#1{equation~\ref{#1}}
\def\Eqref#1{Equation~\ref{#1}}
\def\1{\bm{1}}

\def\vh{{\bm{h}}}

\def\vm{{\bm{m}}}

\def\vx{{\bm{x}}}

\def\mA{{\bm{A}}}
\def\mB{{\bm{B}}}

\def\mE{{\bm{E}}}

\def\mH{{\bm{H}}}
\def\mI{{\bm{I}}}

\def\mM{{\bm{M}}}

\def\mQ{{\bm{Q}}}

\def\mX{{\bm{X}}}

\DeclareMathAlphabet{\mathsfit}{\encodingdefault}{\sfdefault}{m}{sl}
\SetMathAlphabet{\mathsfit}{bold}{\encodingdefault}{\sfdefault}{bx}{n}

\def\gG{{\mathcal{G}}}

\newcommand{\R}{\mathbb{R}}

\usepackage{caption}
\usepackage{babel}
\usepackage{enumitem}
\usepackage{amsthm}
\definecolor{antiquebrass}{rgb}{0.8, 0.58, 0.46}
\definecolor{antiquefuchsia}{rgb}{0.57, 0.36, 0.51}
\usepackage{bbm}
\definecolor{citecol}{HTML}{6F130C}
\definecolor{tableofcontent}{HTML}{1F4A83}
\definecolor{urlcol}{HTML}{2470D8}

\usepackage{hyperref}
\hypersetup{
    colorlinks=true,       % false: boxed links; true: colored links
    linkcolor=tableofcontent,          % color of internal links (change box color with linkbordercolor)
    citecolor=citecol,        % color of links to bibliography
    urlcolor=urlcol,           % color of external links
}

\newcommand{\eg}{\textit{e}.\textit{g}., }
\newcommand{\ie}{\textit{i}.\textit{e}., }

\newcounter{ygincomm}

\definecolor{chromeyellow}{rgb}{1.0, 0.65, 0.0}
\definecolor{mygreen}{rgb}{0.1, 0.8, 0.3}

\definecolor{chromeyellow}{rgb}{1.0, 0.65, 0.0}
\newcounter{bxincomm}
\definecolor{aqua}{rgb}{0.00,0.67,0.80}

\newcounter{todocomm}

\newtheorem{prop}{Propsition}

\newcommand{\ddif}{\textsc{GraDe-IF}}
\title{Graph Denoising Diffusion for Inverse Protein Folding}
\author{
  Kai Yi  $^*$\\
  University of New South Wales \\
  \texttt{kai.yi@unsw.edu.au}
  \And
  Bingxin Zhou \thanks{equal contribution.} \\
  Shanghai Jiao Tong University \\
  \texttt{bingxin.zhou@sjtu.edu.cn} \\
  \And 
  Yiqing Shen\\
  Johns Hopkins University\\
  \texttt{yshen92@jhu.edu}\\
  \\
  \And 
  Pietro Li\`{o} \\
  University of Cambridge\\
  \texttt{\rm pl219@cam.ac.uk}\\
  \And
  Yu Guang Wang \\
  Shanghai Jiao Tong University \\
  University of New South Wales\\
  \texttt{yuguang.wang@sjtu.edu.cn} \\
}

\begin{document}

\maketitle

\begin{abstract}
  Inverse protein folding is challenging due to its inherent one-to-many mapping characteristic, where numerous possible amino acid sequences can fold into a single, identical protein backbone. This task involves not only identifying viable sequences but also representing the sheer diversity of potential solutions. 
  However, existing discriminative models, such as transformer-based auto-regressive models, struggle to encapsulate the diverse range of plausible solutions. 
  In contrast, diffusion probabilistic models, as an emerging genre of generative approaches, offer the potential to generate a diverse set of sequence candidates for determined protein backbones. We propose a novel graph denoising diffusion model for inverse protein folding, where a given protein backbone guides the diffusion process on the corresponding amino acid residue types. The model infers the joint distribution of amino acids conditioned on the nodes' physiochemical properties and local environment. Moreover, we utilize amino acid replacement matrices for the diffusion forward process, encoding the biologically meaningful prior knowledge of amino acids from their spatial and sequential neighbors as well as themselves, which reduces the sampling space of the generative process. Our model achieves state-of-the-art performance over a set of popular baseline methods in sequence recovery and 
  exhibits great potential in generating diverse protein sequences for a determined protein backbone structure. The code is available on \url{https://github.com/ykiiiiii/GraDe_IF}. 
\end{abstract}

\section{Introduction}
Inverse protein folding, or inverse folding, aims to predict feasible amino acid (AA) sequences that can fold into a specified 3D protein structure \cite{khoury2014protein}. 
The results from inverse folding can facilitate the design of novel proteins with desired structural and functional characteristics. 
These proteins can serve numerous applications, ranging from targeted drug delivery to enzyme design for both academic and industrial purposes \cite{madani2023progen,pearce2021deep,sledz2018protein,zhou2023conditional}. In this paper, we develop a diffusion model tailored for graph node denoising to obtain new AA sequences given a protein backbone.

\begin{figure}[!t]
    \centering
    \includegraphics[width=\textwidth]{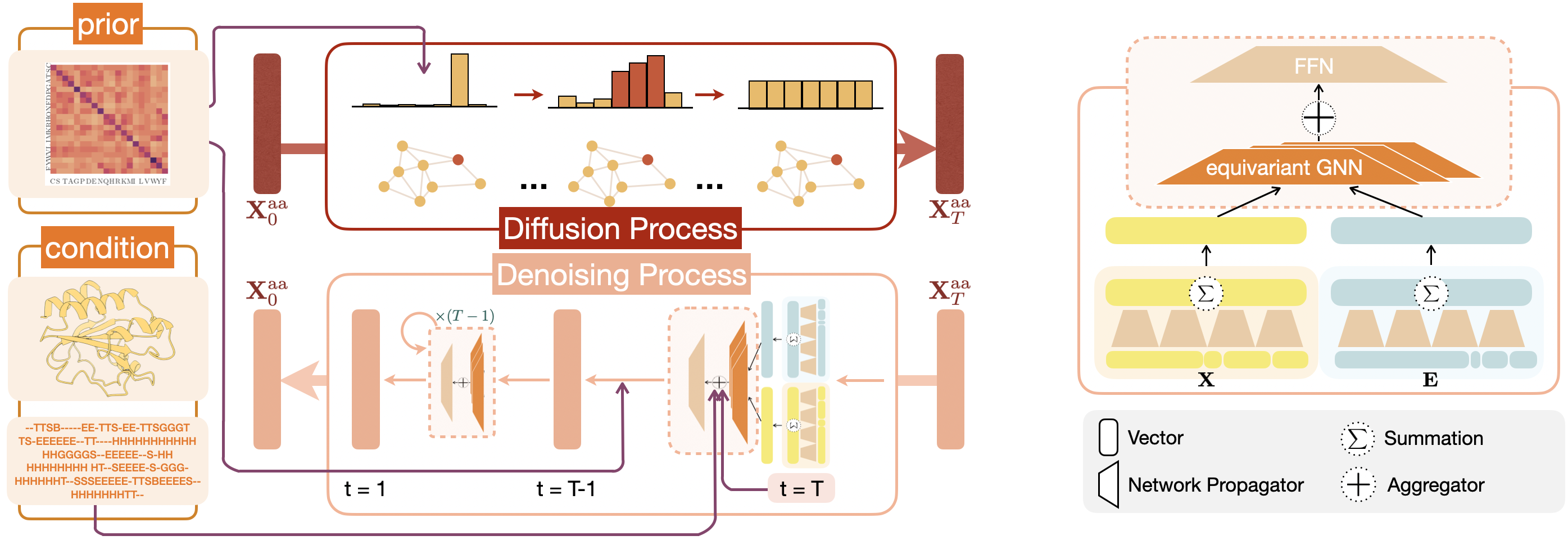}
    \caption{Overview of \ddif. In the diffusion process, the original amino acid is stochastically transitioned to other amino acids, leveraging BLOSUM with varied temperatures as the transition kernel. During the denoising generation phase, initial node features are randomly sampled across the 20 amino acids with a uniform distribution. This is followed by a gradual denoising process, conditional on the graph structure and protein secondary structure at different time points. We employ a roto-translation equivariant graph neural network as the denoising network.
    }
    \vspace{-5mm}
    \label{fig:framework}
\end{figure}

Despite its importance, inverse folding remains challenging due to the immense sequence space to explore, coupled with the complexity of protein folding.
On top of energy-based physical reasoning of a protein’s folded state \cite{alford2017rosetta}, recent advancements in deep learning yield significant progress in learning the mapping from protein structures to AA sequences directly. 
For example, discriminative models formulate this problem as the prediction of the most likely sequence for a given structure via Transformer-based models \cite{ahmed2021prottrans,hsu2022esmif1,lin2023esm2,rao2021msa}. 
However, they have struggled to accurately capture the one-to-many mapping from the protein structure to non-unique AA sequences. 

Due to their powerful learning ability, \textit{diffusion probabilistic models} have gained increasing attention. They are capable of generating a diverse range of molecule outputs from a fixed set of conditions given the inherent stochastic nature. For example, Torsion Diffusion \cite{jing2022torsional} learns the distribution of torsion angles of heavy atoms to simulate conformations for small molecules. 
Concurrently, \textsc{SMCDiff} \cite{trippe2023diffusion} enhances protein folding tasks by learning the stable scaffold distribution supporting a target motif with diffusion. 
Similarly, \textsc{DiffDock} \cite{corso2023diffdock} adopts a generative approach to protein-ligand docking, creating a range of possible ligand binding poses for a target pocket structure. 

Despite the widespread use of diffusion models, their comprehensive potential within the context of protein inverse folding remains relatively unexplored. 
Current methods in sequence design are primarily anchored in language models, encompassing \textit{Masked Language Models} (MLMs) \cite{madani2023progen,lin2023esm2} and \textit{autoregressive generative models} \cite{hsu2022esmif1,meier2021esm1v,nijkamp2022progen2}. 
By tokenizing AAs, MLMs formulate the sequence generation tasks as masked token enrichment. 
These models usually operate by drawing an initial sequence with a certain number of tokens masked as a specific schedule and then learning to predict the masked tokens from the given context. Intriguingly, this procedure can be viewed as a discrete diffusion-absorbing model when trained by a parameterized objective. Autoregressive models, conversely, can be perceived as deterministic diffusion processes \cite{austin2021structured}. It induces conditional distribution to each token, but the overall dependency along the entire AA sequence is recast via an independently-executed diffusion process. 

On the contrary, diffusion probabilistic models employ an iterative prediction methodology that generates less noisy samples and demonstrates potential in capturing the diversity inherent in real data distributions. This unique characteristic further underscores the promising role diffusion models could play in advancing the field of protein sequence design.
To bridge the gap, we make the first attempt at a diffusion model for inverse folding. We model the inverse problem as a denoising problem where the randomly assigned AA types in a protein (backbone) graph is recovered to the wild type. The 
protein graph which contains the spatial and biochemical information of all AAs is represented by equivariant graph neural networks, and diffusion process takes places on graph nodes. In real inverse folding tasks, the proposed model achieves SOTA recovery rate, improve  4.2\% and 5.4\% on recovery rate for single-chain proteins and short sequences, respectively,
, especially for conserved region which has a biologically significance. Moreover, the predicted structure of generated sequence is identical to the structure of native sequence.

The preservation of the desired functionalities is achieved by innovatively conditioning the model on both secondary and third structures in the form of residue graphs and corresponding node features.
The major contributions of this paper are three-fold. Firstly, we propose \ddif, a diffusion model backed by roto-translation equivariant graph neural network for inverse folding. It stands out from its counterparts for its ability to produce a wide array of diverse sequence candidates. Secondly, as a departure from conventional uniform noise in discrete diffusion models, we encode the prior knowledge of the response of AAs to evolutionary pressures by the utilization of \textit{Blocks Substitution Matrix} as the translation kernel. Moreover, to accelerate the sampling process, we adopt Denoising Diffusion Implicit Model (DDIM) from its original continuous form to suit the discrete circumstances and back it with thorough theoretical analysis.

% \section{Discrete Denoising Diffusion for Inverse Folding} 
\section{Problem Formulation}
%problem setting

\subsection{Residue Graph by Protein Backbone}
A residue graph, denoted as $\gG=(\mX, \mA, \mE)$, aims to delineate the geometric configuration of a protein. Specifically, every node stands for an AA within the protein. Correspondingly, each node is assigned a collection of meticulously curated node attributes $\mX$ to reflect its physiochemical and topological attributes. 
The local environment of a given node is defined by its spatial neighbors, as determined by the $k$-nearest neighbor ($k$NN) algorithm. 
Consequently, each AA node is linked to a maximum of $k$ other nodes within the graph, specifically those with the least Euclidean distance amongst all nodes within a 30\AA~contact region. 
The edge attributes, represented as $\mE\in\R^{93}$, illustrate the relationships between connected nodes. These relationships are determined through parameters such as inter-atomic distances, local N-C positions, and a sequential position encoding scheme. We detail the attribute construction in Appendix~\ref{sec:app:proteinGraph}.

\subsection{Inverse Folding as a Denoising Problem}
The objective of inverse folding is to engineer sequences that can fold to a pre-specified desired structure. We utilize the coordinates of C$\alpha$ atoms to represent the 3D positions of AAs in Euclidean space, thereby embodying the protein backbone. Based on the naturally existing protein structures, our model is constructed to generate a protein's native sequence based on the coordinates of its backbone atoms. Formally we represent this problem as learning the conditional distribution $p(\mX^{\rm aa}|\mX^{\rm pos})$. Given a protein of length $n$ and a sequence of spatial coordinates $\mX^{\rm pos}=\{\vx^{\rm pos}_1,\dots,\vx^{\rm pos}_i,\dots,\vx^{\rm pos}_{n}\}$ representing each of the backbone C$\alpha$ atoms in the structure, the target is to predict $\mX^{\rm aa}=\{\vx^{\rm aa}_1,\dots,\vx^{\rm aa}_i,\dots,\vx^{\rm aa}_n\}$, the native sequence of AAs. This density is modeled in conjunction with the other AAs along the entire chain. Our model is trained by minimizing the negative log-likelihood of the generated AA sequence relative to the native wild-type sequence. Sequences can then be designed either by sampling or by identifying sequences that maximize the conditional probability given the desired secondary and tertiary structure.

\subsection{Discrete Denoising Diffusion Probabilistic Models}
Diffusion models belong to the class of generative models, where the training stage encompasses diffusion and denoising processes. The diffusion process $q\left(\vx_1,\dots,\vx_{T} \mid \vx_0 \right)=\prod_{t=1}^T q\left(\vx_{t} \mid \vx_{t-1}\right)$ corrupts the original data $\vx_0 \sim$ $q\left(\vx\right)$ into a series of latent variables $\{\vx_1,\dots,\vx_{T}\}$, with each carrying progressively higher levels of noise. 
Inversely, the denoising process $p_\theta\left(\vx_0,\vx_1,...,\vx_T\right)=p\left(\vx_T\right) \prod_{t=1}^T p_\theta\left(\vx_{t-1} \mid \vx_t\right)$ gradually reduces the noise within these latent variables, steering them back towards the original data distribution. The iterative denoising procedure is driven by a differentiable operator, such as a trainable neural network.

While in theory there is no strict form for $q\left(\vx_{t} \mid \vx_{t-1}\right)$ to take, several conditions are required to be fulfilled by $p_\theta$ for efficient sampling: (i) The diffusion kernel $q(\vx_{t}|\vx_0)$ requires a closed form to sample noisy data at different time steps for parallel training. (ii) The kernel should possess a tractable formulation for the posterior $q\left(\vx_{t-1} \mid \vx_{t}, \vx_0\right)$. Consequently, the posterior $p_\theta(\vx_{t-1} | \vx_{t}) = \int q\left(\vx_{t-1} \mid \vx_{t}, \vx_0\right) {\rm d} p_\theta(\vx_0|\vx_{t})$, and $\vx_0$ can be used as the target of the trainable neural network.
(iii) The marginal distribution $q(\vx_{T})$ should be independent of $\vx_0$. This independence allows us to employ $q(\vx_{T})$ as a prior distribution for inference.

The aforementioned criteria are crucial for the development of suitable noise-adding modules and training pipelines. To satisfy these prerequisites, we follow the setting in previous work \cite{austin2021structured}. For categorical data $\vx_t \in \{1,...,K\}$, the transition probabilities are calculated by the matrix $\left[\mQ_{t}\right]_{ij}=$ $q\left(\vx_t=j \mid \vx_{t-1}=i\right)$. 
% \ygc{check the notation and meaning of the variables $q$, $x_t$} 
Employing the transition matrix and on one-hot encoded categorical feature $\vx_t$, we can define the transitional kernel in the diffusion process by:
\begin{equation}
    q\left(\vx_{t} \mid \vx_{t-1}\right)=\vx_{t-1} \mQ_{t} \quad \text { and } \quad q\left(\vx_{t} \mid \vx\right)=\vx \bar{\mQ}_t,
\end{equation}
where $\bar{\mQ}_t=\mQ_1 \ldots \mQ_{t}$. 
The Bayes rule yields that the posterior distribution can be calculated in closed form as $q\left(\vx_{t-1} \mid \vx_{t}, \vx\right) \propto \vx_{t}\mQ_{t}^{\top} \odot \vx\bar{\mQ}_{t-1}$. The generative probability can thus be determined using the transition kernel, the model output at time $t$, and the state of the process $\vx_t$. Through iterative sampling, we eventually produce the generated output $\vx_0$.

The prior distribution $p(\vx_T)$ should be independent of the observation $\vx_0$. Consequently, the construction of the transition matrix necessitates the use of a noise schedule. The most straightforward and commonly utilized method is the uniform transition, which can be parameterized as $\mQ_t = \alpha_t \mI + (1 - \alpha_t)\1_d\1_d^{\top}/d$ with $\mI^{\top}$ be the transpose of the identity matrix $\mI$, $d$ refers to the number of amino acid types (\ie $d=20$) and $\1_d$ denotes the one vector of dimension $d$. As $t$ approaches infinity, $\alpha$ undergoes a progressive decay until it reaches $0$. Consequently, the distribution $q(\vx_T)$ asymptotically approaches a uniform distribution, which is essentially independent of $\vx$.

\section{Graph Denoising Diffusion for Inverse Protein Folding}

In this section, we introduce a discrete graph denoising diffusion model for protein inverse folding, which utilizes a given graph $\gG = \{\mX,\mA,\mE\}$ with node feature $\mX$ and edge feature $\mE$ as the condition. 
Specifically, the node feature depicts the AA position, AA type, and the spatial and biochemical properties $\mX = [\mX^{\rm pos},\mX^{\rm aa}, \mX^{\rm prop}]$.
We define a diffusion process on the AA feature $\mX^{\rm aa}$, and denoise it conditioned on the graph structure $\mE$ which is encoded by \emph{equivariant neural networks} \cite{satorras2021n}. Moreover, we incorporate protein-specific prior knowledge, including an \emph{AA substitution scoring matrix} and protein \emph{secondary structure} during modeling. We also introduce a new acceleration algorithm for the discrete diffusion generative process based on a transition matrix.

\subsection{Diffusion Process and Generative Denoising Process}

\paragraph{Diffusion Process}
To capture the distribution of AA types, we independently add noise to each AA node of the protein. For any given node, the transition probabilities are defined by the matrix $\mQ_{t}$.  With the predefined transition matrix, we can define the forward diffusion kernel by
$$
q\left(\mX^{\rm aa}_t \mid \mX^{\rm aa}_{t-1}\right)=\mX^{\rm aa}_{t-1} \mQ_{t} \quad \text { and } \quad q\left(\mX^{\rm aa}_t \mid \mX^{\rm aa}\right)=\mX^{\rm aa} \bar{\mQ}_t,
$$
where $\bar{\mQ}_t=\mQ_1 \ldots \mQ_{t}$ is the transition probability matrix up to step $t$. 

\paragraph{Training Denoising Networks}
The second component of the diffusion model is the denoising neural network $f_\theta$, parameterized by $\theta$. This network accepts a noisy input $\gG_t=\left(\mX_{t}, \mathbf{E}\right)$, where $\mX_{t}$ is the concatenation of the noisy AA types and other AA properties including 20 one-hot encoded AA type and 15 geometry properties, such as SASA, normalized surface-aware node features, dihedral angles of backbone atoms, and 3D positions.
It aims to predict the clean type of AA $\mX^{\rm aa}$, which allows us to model the underlying sequence diversity in the protein structure while maintaining their inherent structural constraints.   To train $f_\theta$, we optimize the cross-entropy loss $L$ between the predicted probabilities $\hat{p}(\mX^{\rm aa})$ for each node's AA type. 

\paragraph{Parameterized Generative Process}
A new AA sequence is generated through the reverse diffusion iterations on each node $\vx$.  The generative probability distribution $p_\theta(\vx_{t-1}|\vx_{t})$ is estimated from the predicted probability $\hat{p}(\vx^{\rm aa}|\vx_t)$ by the neural networks. We marginalize over the network predictions to compute for generative distribution at each iteration:
\begin{equation}
\label{eq:generativeProb}
  p_\theta\left(\vx_{t-1} \mid \vx_t\right) \propto\sum_{\hat{\vx}^{\rm aa}}q(\vx_{t-1}|\vx_t,\vx^{\rm aa})\hat{p}_{\theta}(\vx^{\rm aa}|\vx_t),
\end{equation}
%=\int_{x_0} q\left(x_{t-1} \mid \hat{x}_0, x_t\right) d \hat{p}_\theta\left(x_0 \mid x_t\right)
where the posterior %$q\left(\vx_{t-1} \mid  \vx_t,\vx^{\rm aa}\right)$
\begin{equation}
\label{eq:posteriorDist}
   q\left(\vx_{t-1} \mid  \vx_t,\vx^{\rm aa}\right) = \text{Cat}\left(\vx_{t-1}\Big|\frac{\vx_tQ_t^{\top} \odot \vx^{\rm aa}\bar{Q}_{t-1}}{\vx^{\rm aa}\bar{Q}_t\vx_t^{\top}}\right)
\end{equation}
can be calculated from the transition matrix, state of node feature at step $t$ and AA type $\vx^{\rm aa}$. The $\vx^{\rm aa}$ is the sample of the denoising network prediction $\hat{p}(\vx^{\rm aa})$. 

% \subsection{Adding Noise}
% For any node, the transit probabilities are defined by the matrices $\left[\mQ_t\right] = (1-\beta_t)\boldsymbol{I} + \beta_t/20$. By adding noise simultaneously on every node, we get the conditional probabilities for the noisy node:
% $$q(x_t|x_{t-1}) = x_{t-1}Q_t \quad \text{and} \quad q(x_t|x_0) = x_0\bar{Q}_t$$
% for $\bar{Q}_t = Q_1Q_2 \dots Q_t\\$

% \yg{TODO: distinguish scalar and vector}

% \subsection{Denoising Network}
% The denoising neural network $f_\theta$ takes a noisy graph node feature $X_t$ and tries to predict the clean node feature $X$. To train it, we optimized the cross-entropy $l$ between predicted probabilities $\hat{p}(x_0)$ for each node and the true node feature $x$
% $$l(\hat{p}(x),x) = -\sum_{1\leq n \leq N}\sum_{1\leq i \leq 20}x_i\log \hat{p}_{\theta}(x_i)$$

% \subsection{Transition Matrix}

\subsection{Prior Distribution from Protein Observations}

\subsubsection{Markov Transition Matrices}
The transition matrix serves as a guide for a discrete diffusion model, facilitating transitions between the states by providing the probability of moving from the current time step to the next. As it reflects the possibility from one AA type to another, this matrix plays a critical role in both the diffusion and generative processes. During the diffusion stage, the transition matrix is iteratively applied to the observed data, which evolves over time due to inherent noise. As diffusion time increases, the probability of the original AA type gradually decays, eventually converging towards a uniform distribution across all AA types. In the generative stage, the conditional probability $p(\vx_{t-1}|\vx_t)$ is determined by both the model's prediction and the characteristics of the transition matrix $\mQ$, as described in \Eqref{eq:generativeProb}. 

Given the biological specificity of AA substitutions, the transition probabilities between AAs are not uniformly distributed, making it illogical to define random directions for the generative or sampling process. As an alternative, the diffusion process could reflect evolutionary pressures by utilizing substitution scoring matrices that conserve protein functionality, structure, or stability in wild-type protein families.
Formally, an \emph{AA substitution scoring matrix} quantifies the rates at which various AAs in proteins are substituted by other AAs over time \cite{trivedi2020substitution}. In this study, we employ the Blocks Substitution Matrix (BLOSUM) \cite{henikoff1992blosum}, which identifies conserved regions within proteins that are presumed to have greater functional relevance. Grounded in empirical observations of protein evolution, BLOSUM provides an estimate of the likelihood of substitutions between different AAs. We thus incorporate BLOSUM into both the diffusion and generative processes. Initially, the matrix is normalized into probabilities using the softmax function. Then, we use the normalized matrix $\mB$ with different probability temperatures to control the noise scale of the diffusion process. Consequently, the transition matrix at time $t$ is given by $\mQ_t = \mB^{T}$. By using this matrix to refine the transition probabilities, the generative space to be sampled is reduced effectively, thereby the model's predictions converge toward a meaningful subspace. See Figure~\ref{fig:blosumSampling} for a comparison of the transition matrix over time in random and \textsc{BLOSUM} cases.

\begin{figure}[t]
    \centering
    \includegraphics[width=\textwidth]{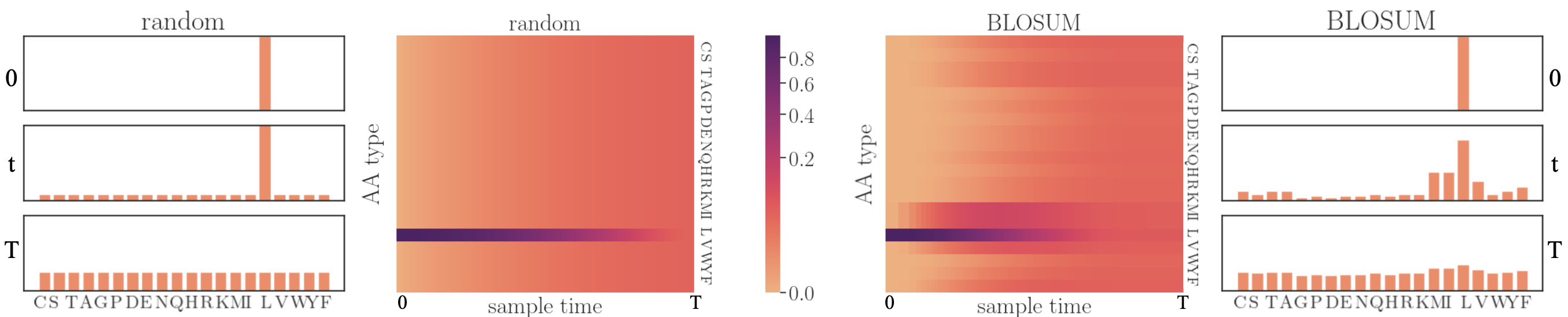}
    \caption{The middle two panels depict the transition probability of Leucine (L) from $t=0$ to $T$. Both the uniform and BLOSUM start as Dirichlet distributions and become uniform at time $T$. As shown in the two side figures, while the uniform matrix evenly disperses L's probability to other AAs over time, BLOSUM favors AAs similar to L. }
    \label{fig:blosumSampling}
    \vspace{-0.6cm}
\end{figure}

\subsubsection{Secondary Structure}
Protein secondary structure refers to the local spatial arrangement of AA residues in a protein chain. The two most common types of protein secondary structure are alpha helices and beta sheets, which are stabilized by hydrogen bonds between backbone atoms. The secondary structure of a protein serves as a critical intermediary, bridging the gap between the AA sequence and the overall 3D conformation of the protein. In our study, we incorporate eight distinct types of secondary structures into AA nodes as conditions during the sampling process. This strategic approach effectively narrows down the exploration space of potential AA sequences. Specifically, we employ DSSP (Define Secondary Structure of Proteins) to predict the secondary structures of each AA and represent these structures using one-hot encoding. Our neural network takes the one-hot encoding as input and utilizes it to denoise the AA conditioned on it.

The imposition of motif conditions such as alpha helices and beta sheets on the search for AA sequences not only leads to a significant reduction in the sampling space of potential sequences, but also imparts biological implications for the generated protein sequence. By conditioning the sampling process of AA types on their corresponding secondary structure types, we guide the resulting protein sequence towards acquiring not only the appropriate 3D structure with feasible thermal stability but also the capability to perform its intended function.

\subsection{Equivariant Graph Denoising Network}
Bio-molecules such as proteins and chemical compounds are structured in the 3-dimensional space, and it is vital for the model to predict the same binding complex no matter how the input proteins are positioned and oriented to encode a robust and expressive hidden representation. This property can be guaranteed by the rotation equivariance of the neural networks. A typical choice of such a network is an equivariant graph neural network \cite{satorras2021n}.
We modify its SE(3)-equivariant neural layers to update representations for both nodes and edges, which reserves SO(3) rotation equivariance and E(3) translation invariance. At the $l$th layer, an Equivariant Graph Convolution (\textsc{EGC}) inputs a set of $n$ hidden node embeddings 
$\mH^{(l)}=\left\{\vh_1^{(l)},\dots,\vh_n^{(l)}\right\}$ describing AA type and geometry properties, edge embedding $\vm_{ij}^{(l)}$ with respect to connected nodes $i$ and $j$, $\mX^{\rm pos}=\left\{\vx_1^{\rm pos}, \dots, \vx_n^{\rm pos}\right\}$ for node coordinates and $t$ for time step embedding of diffusion model. The target of a modified \textsc{EGC} layer is to update hidden representations $\mH^{(l+1)}$ for nodes and $\mM^{(l+1)}$ for edges. Concisely, $\mH^{(l+1)}, \mM^{(l+1)}=\operatorname{EGC}\left[\mH^{(l)}, \mX^{\text{pos}}, \mM^{(l)}, t \right]$. To achieve this, an \textsc{EGC} layer defines 
\begin{equation}
\label{eq:egnn}
    \begin{aligned}
    \vm_{ij}^{(l+1)} &=\phi_{e}\left(\mathbf{h}_{i}^{(l)}, \mathbf{h}_{j}^{(l)},\left\|\mathbf{x}_{i}^{(l)}-\mathbf{x}_{j}^{(l)}\right\|^{2}, \vm^{(l)}_{ij}\right) \\
    \vx_{i}^{(l+1)} &=\mathbf{x}_{i}^{(l)}+\frac1{n}\sum_{j \neq i}\left(\mathbf{x}_{i}^{(l)}-\mathbf{x}_{j}^{(l)}\right) \phi_{x}\left(\mathbf{m}_{i j}^{(l+1)}\right) \\
    \vh_{i}^{(l+1)} &=\phi_{h}\big(\mathbf{h}_{i}^{(l)}, \sum_{j \neq i} \mathbf{m}_{i j}^{(l+1)}\big),
\end{aligned}
\end{equation}
where $\phi_e, \phi_h$ are the edge and node propagation operations, respectively. The $\phi_x$ is an additional operation that projects the vector edge embedding $\vm_{ij}$ to a scalar. 
% The modified EGC layer starts from aggregating representations of node pairs with their edge attributes and the Euclidean distance between the nodes. Next, the nodes' 3D positions for the next layer are updated with the projected propagated embedding ($\phi_x(\vm_{ij})$) as well as the differences in the coordinates of neighboring nodes within the 1-hop range. In the final third step, the hidden embedding for the node $i$ is updated by a conventional message passing of node $i$ and its 1-hop neighbors' hidden embedding from the previous steps. 
The modified EGC layer preserves equivariance to rotations and translations on the set of 3D node coordinates $\mX^{\text{pos}}$ and performs invariance to permutations on the nodes set identical to any other GNNs.

\subsection{DDIM Sampling Process}
A significant drawback of diffusion models lies in the speed of generation process, which is typically characterized by numerous incremental steps and can be quite slow. Deterministic Denoising Implicit Models (DDIM) \cite{songdenoising} are frequently utilized to counter this issue in continuous variable diffusion generative models. DDIM operates on a non-Markovian forward diffusion process, consistently conditioning on the input rather than the previous step. By setting the noise variance on each step to $0$, the reverse generative process becomes entirely deterministic, given an initial prior sample.

Similarly, since we possess the closed form of generative probability $p_{\theta}(\vx_{t-1}|\vx_t)$ in terms of a predicted $\vx^{\rm aa}$ and the posterior distribution $p(\vx_{t-1}|\vx_t,\vx^{\rm aa})$, we can also render the generative model deterministic by controlling the sampling temperature of $p(\vx^{\rm aa}|\vx_t)$. Consequently, we can define the multi-step generative process by 
\begin{equation}
\label{ddim}
\begin{aligned}
p_\theta\left(\vx_{t-k} \mid \vx_t\right)
\propto (\sum_{\hat{\vx}^{\rm aa}}q(\vx_{t-k}|\vx_t,\vx^{\rm aa})\hat{p}(\vx^{\rm aa}|\vx_t))^T    
\end{aligned}
\end{equation}

where the temperature $T$ controls whether it is deterministic or stochastic, and the multi-step posterior distribution is
\begin{equation}
\label{eq:ddimPosterior}
   q\left(\vx_{t-k} \mid  \vx_t,\vx^{\rm aa}\right) = \text{Cat}\left(\vx_{t-k}\Big|\frac{\vx_t Q_t^{\top}\cdots Q_{t-k}^{\top } \odot \vx^{\rm aa}\bar{Q}_{t-k}}{\vx^{\rm aa}\bar{Q}_t\vx_t^{\top}}\right). 
\end{equation}

% \subsection{Algorithms}

\section{Experiments}
We validate our \ddif~on recovering native protein sequences in \textbf{CATH} \cite{ORENGO1997cath}. The performance is mainly compared with structure-aware SOTA models. The implementations for the main algorithms (see Appendix~\ref{sec:app:algorithm}) at \url{https://github.com/ykiiiiii/GraDe_IF} are programmed with \texttt{PyTorch-Geometric} (ver 2.2.0) and \texttt{PyTorch} (ver 1.12.1) and executed on an NVIDIA$^{\circledR}$ Tesla V100 GPU with $5,120$ CUDA cores and $32$GB HBM2 installed on an HPC cluster.

\subsection{Experimental Protocol}
\paragraph{Training Setup}
We employ \textbf{CATH v4.2.0}-based partitioning as conducted by \textsc{GraphTrans} \cite{ingraham2019generative} and GVP \cite{jing2021gvp}. Proteins are categorized based on \textbf{CATH} topology classification, leading to a division of $18,024$ proteins for training, $608$ for validation, and $1,120$ for testing. To evaluate the generative quality of different proteins, we test our model across three distinct categories: \textit{short}, \textit{single-chain}, and \textit{all} proteins. The short category includes proteins with sequence lengths shorter than 100. The single-chain category encompasses proteins composed of a single chain. In addition, the total time step of the diffusion model is configured as $500$, adhering to a cosine schedule for noise \cite{nichol2021improved}. For the denoising network, we implement six stacked EGNN blocks, each possessing a hidden dimension of $128$. Our model undergoes training for default of $200$ epochs, making use of the Adam optimizer. A batch size of $64$ and a learning rate of $0.0005$ are applied during training. Moreover, to prevent overfitting, we incorporate a dropout rate of 0.1 into our model's architecture.

\paragraph{Evaluation Metric}
The quality of recovered protein sequences is quantified by \textit{perplexity} and \textit{recovery rate}. The former measures how well the model's predicted AA probabilities match the actual AA at each position in the sequence. A lower perplexity indicates a better fit of the model to the data. The recovery rate assesses the model's ability to recover the correct AA sequence given the protein's 3D structure. It is typically computed as the proportion of AAs in the predicted sequence that matches the original sequence. A higher recovery rate indicates a better capability of the model to predict the original sequence from the structure.

\begin{table}[!t]
\caption{Recovery rate performance of \textbf{CATH} on zero-shot models.}
\label{tab:rr}
\begin{center}
\resizebox{\textwidth}{!}{
    \begin{tabular}{ccccccccc}
    \toprule
    \multirow{2}{*}{\textbf{Model}} & \multicolumn{3}{c}{ \textbf{Perplexity} $\downarrow$} & \multicolumn{3}{c}{ \textbf{Recovery Rate} \% $\uparrow$} & \multicolumn{2}{c}{\textbf{CATH version}} \\\cmidrule(lr){2-4}\cmidrule(lr){5-7}\cmidrule(lr){8-9}
    & Short & Single-chain & All & Short & Single-chain & All & 4.2 & 4.3 \\
    \midrule 
    \textsc{StructGNN} \cite{ingraham2019generative} & 8.29 & 8.74 & 6.40 & 29.44 & 28.26 & 35.91 & $\checkmark$ & \\
    \textsc{GraphTrans} \cite{ingraham2019generative} & 8.39 & 8.83 & 6.63 & 28.14 & 28.46 & 35.82 & $\checkmark$ & \\
    GCA \cite{tan2022generative} & 7.09 & 7.49 & 6.05 & 32.62 & 31.10 & 37.64 & $\checkmark$ & \\
    GVP \cite{jing2021gvp} & 7.23 & 7.84 & 5.36 & 30.60 & 28.95 & 39.47 & $\checkmark$ & \\
    GVP-large \cite{hsu2022esmif1} & 7.68 & 6.12 & 6.17 & 32.6 & 39.4 & 39.2 & & $\checkmark$ \\
    \textsc{AlphaDesign} \cite{gao2022alphadesign} & 7.32 & 7.63 & 6.30 & 34.16 & 32.66 & 41.31 & $\checkmark$ & \\
    \textsc{ESM-if1} \cite{hsu2022esmif1} & 8.18 & 6.33 & 6.44 & 31.3 & 38.5 & 38.3 & & $\checkmark$ \\
    \textsc{ProteinMPNN} \cite{dauparas2022robust} & 6.21 & 6.68 & 4.57 & 36.35 & 34.43 & 49.87 & $\checkmark$ & \\
    \textsc{PiFold} \cite{gao2023pifold} & $6.04$ & 6.31 & $4 . 5 5$ & $39.84$ & 38.53 & 51.66 & $\checkmark$ & \\
    \midrule
    \ddif &$\mathbf{5.49}$&$\mathbf{6.21}$&$\mathbf{4.35}$&$\mathbf{45.27}$&$\mathbf{42.77}$&$\mathbf{52.21}$&$\checkmark$\\
    \bottomrule
    \end{tabular} 
}
\end{center}
\end{table}

\begin{figure}[!t]
    \centering
    \includegraphics[width=\textwidth]{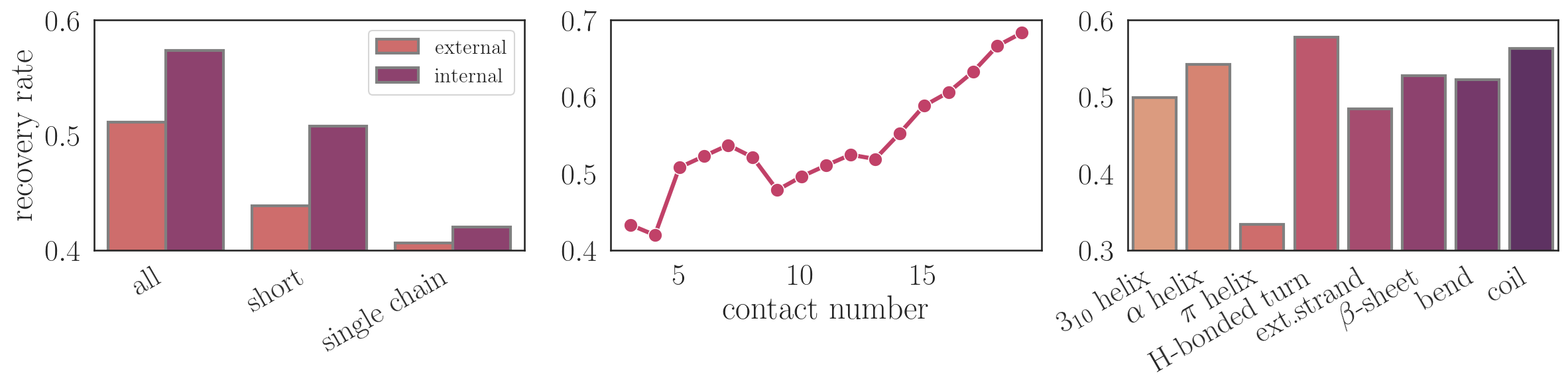}
    \vspace{-8mm}
    \caption{Recovery rate on core and surface residues and different secondary structure }
    \label{fig:recoverRate}
    \vspace{4mm}
\end{figure}

\subsection{Inverse Folding}\label{sec:inversen folding}
Table~\ref{tab:rr} compares \ddif's performance on recovering proteins in \textbf{CATH}, with the last two columns indicating the training dataset of each baseline method. To generate high-confidence sequences, \ddif~integrates out uncertainties in the prior by approximating the probability $p(\vx^{\rm aa}) \approx \sum_{i=1}^N p(\vx^{\rm aa}|\vx^i_T)p(\vx^i_T)$. Notably, we observed an improvement of $4.2\%$ and $5.4\%$ in the recovery rate for single-chain proteins and short sequences, respectively. We also conducted evaluations on different datasets (Appendix~\ref{sec:app:ts50500}) and ablation conditions (Appendix~\ref{sec:app:ablation}).

Upon subdividing the recovery performance based on buried and surface AAs, we find that the more conserved core residues exhibit a higher native sequence recovery rate. In contrast, the active surface AAs demonstrate a lower sequence recovery rate. Figure~\ref{fig:recoverRate} examines AA conservation by Solvent Accessible Surface Area (SASA) (with SASA$<0.25$ indicating internal AAs) and contact number (with the number of neighboring AAs within 8 \AA~in 3D space) \cite{gong2017improving}. The recovery rate of internal residues significantly exceeds that of external residues across all three protein sequence classes, with the recovery rate increasing in conjunction with the contact number. We also present the recovery rate for different secondary structures, where we achieve high recovery for the majority of secondary structures, with the exception of a minor $5$-turn helix structure that occurs infrequently.

\subsection{Foldability}

\begin{table}[t]
\centering
\caption{Numerical comparison between generated sequence structure and the native structure.}
\label{tab:Foldability}
\resizebox{0.7\textwidth}{!}{
    \begin{tabular}{ccccc}
    \toprule
    Method & Success & TM score & avg pLDDT & avg RMSD\\
    \midrule
    \textsc{PiFOLD} & $85$ & $0.80 \pm 0.22$ & $0.84 \pm 0.15$ & $1.67 \pm 0.99$\\
    \textsc{ProteinMPNN} & $94$ & $0.86 \pm 0.16$ & $0.89 \pm 0.10$ & $1.36 \pm 0.81$\\
    \ddif & $94$ & $0.86 \pm 0.17$ & $0.86 \pm 0.08$ & $1.47 \pm 0.82$\\
    \bottomrule
    \end{tabular}
}
\end{table}

\begin{figure}[t]
    \centering
    \includegraphics[width=0.9\textwidth]{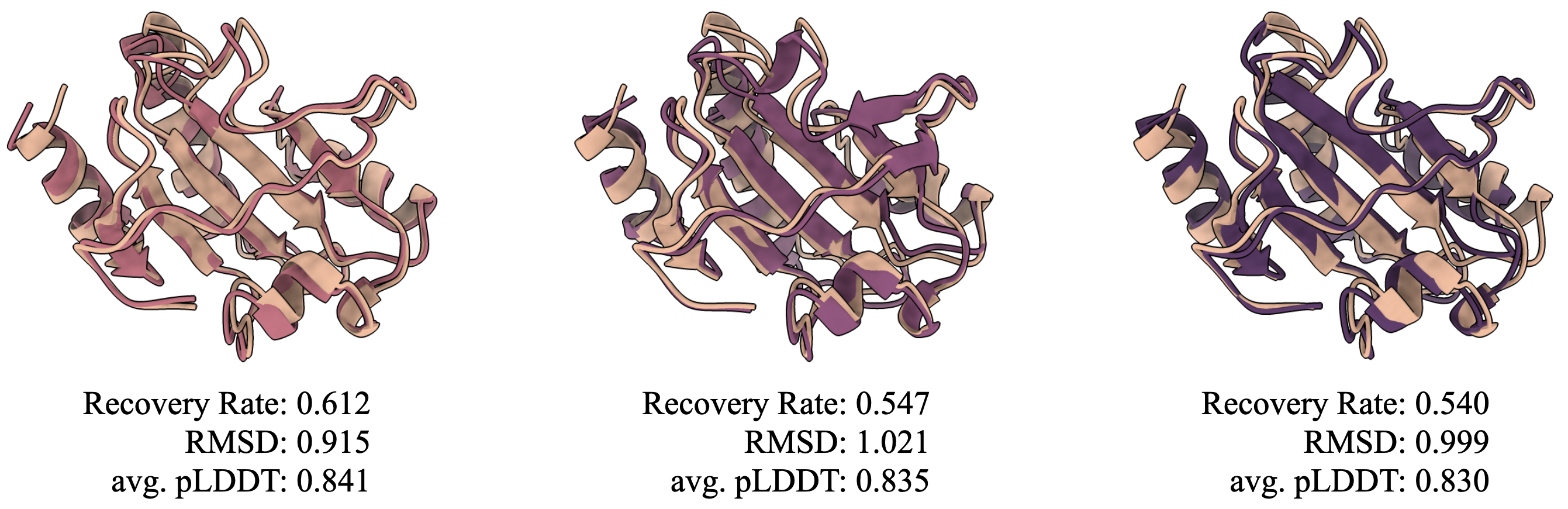}
    \caption{Folding comparison of \ddif-generated sequences and the native protein (in nude).}
    \vspace{-3mm}
    \label{fig:folding}
\end{figure}

We extend our investigation to the foldability of sequences generated at various sequence recovery rates. We fold generated protein sequences (by \ddif) with \textsc{AlphaFold2} and align them with the crystal structure to compare their closeness in Figure~\ref{fig:folding} (PDB ID: 3FKF). All generated sequences are nearly identical to the native one with their RMSD $\sim1$ \AA~over $139$ AAs, which is lower than the resolution of the crystal structure at $2.2$ \AA. We have also folded the native sequence by \textsc{AlphaFold2}, which yields an average pLDDT of $0.91$. In comparison, the average pLDDT scores of the generated sequences are $0.835$, underscoring the reliability of their folded structures. In conjunction with the evidence presented in Figure~\ref{fig:recoverRate} which indicates our method's superior performance in generating more identical results within conserved regions, we confidently posit that \ddif~can generate biologically plausible novel sequences for given protein structures (See Appendix~\ref{sec:app:af2}).

The numerical investigation is reported in Table~\ref{tab:Foldability}, where we pick the first $100$ structures (ordered in alphabetical order by their PDB ID) from the test dataset and compare the performance of \ddif~with \textsc{ProteinMPNN} and \textsc{PiFOLD}. We follow \cite{wu2022protein} and define the quality of a novel sequence by the TM score between its \textsc{AlphaFold2}-folded structure and the native structure, with an above-$0.5$ score indicating the design is \textit{successful}. Overall, our \ddif~exhibits high pLDDT, low RMSD, and high foldability. There are several proteins that face challenges with both \ddif~and baseline methods for folding with a high TM score, \ie 1BCT, 1BHA, and 1CYU, whose structural determination is based on NMR, an experimental technique that analyzes protein structure in a buffer solution. Due to the presence of multiple structures for a single protein in NMR studies, it is reasonable for folding tools to assign low foldability scores.

\begin{table}[t]
\centering
\caption{Numerical comparison on diversity and recovery rate}
\label{tab:diversity}
\resizebox{0.85\textwidth}{!}{
    \begin{tabular}{lcccccc}
    \toprule
    & \multicolumn{2}{c}{\textbf{low recovery rate}} & \multicolumn{2}{c}{\textbf{medium recovery rate}} & \multicolumn{2}{c}{\textbf{high recovery rate}} \\ \cmidrule(lr){2-3} \cmidrule(lr){4-5} \cmidrule(lr){6-7}
    \textbf{Method} & diversity & recovery & diversity & recovery & diversity & recovery \\
    \midrule
    \textsc{PiFold} & 0.37 & 0.47 & 0.25 & 0.50 & 0.21 & 0.50 \\
    \textsc{ProteinMPNN} & 0.51 & 0.42 & 0.27 & 0.46 & 0.26 & 0.52 \\
    \ddif & 0.61 & 0.33 & 0.54 & 0.47 & 0.25 & 0.53 \\
    \bottomrule
    \end{tabular}
}
\end{table}

\begin{figure}[t]
    \centering
    \includegraphics[width=\textwidth]{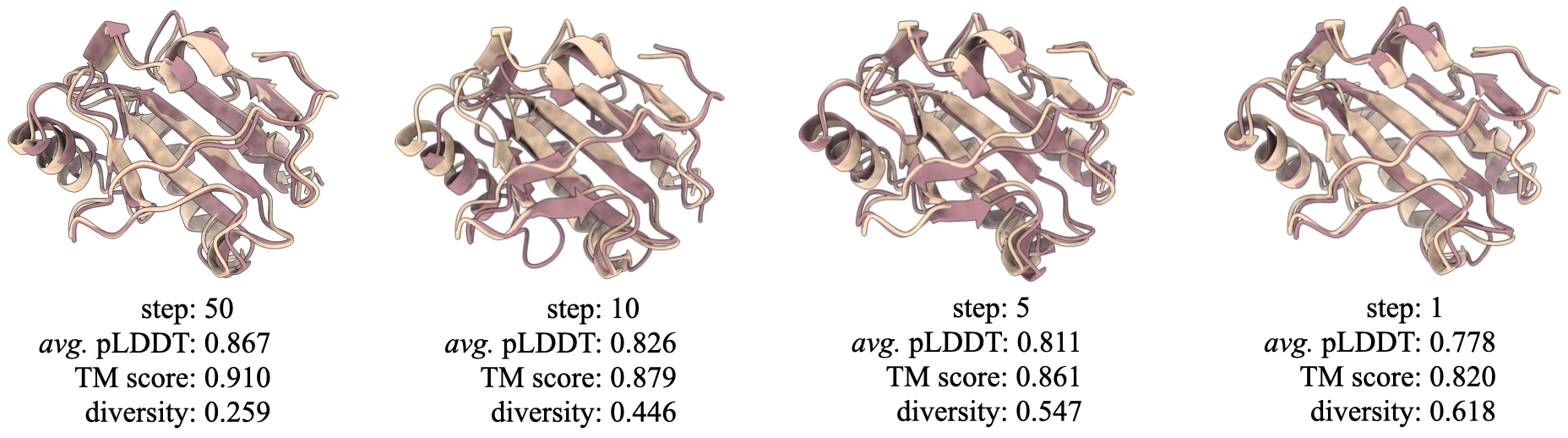}
    \caption{Folding structures of generated protein sequences with different steps. }
    \label{fig:diff_step}
\end{figure}

\subsection{Diversity}
Given the one-to-many mapping relationship between a protein structure and its potential sequences, an inverse folding model must be capable of generating diverse protein sequences for a fixed backbone structure. To investigate the diversity of \ddif~in comparison to baseline methods, we define \emph{diversity} as $1 - \text{self-similarity}$, where higher diversity indicates a model's proficiency in generating distinct sequences. We employ \textsc{PiFold}, \textsc{ProteinMPNN}, and \ddif~to generate new proteins for the test dataset at low, medium, and high recovery levels. We vary the temperature for the former two baseline methods within the range $\{0.5, 0.1, 0.0001\}$ and adjust the sample step for \ddif~from $\{1, 10, 50\}$. The average performance from sampling $10$ sequences is summarized in Table~\ref{tab:diversity}, revealing comparable results among the three models in both recovery rate and diversity. In general, increasing the probability temperature for \textsc{PiFold} and \textsc{ProteinMPNN} (or decreasing the sample step in \ddif) leads to a reduction in uncertainty and more reliable predictions, resulting in higher diversities and lower recovery rates for all three models. While all three methods demonstrate the ability to recover sequences at a similar level, particularly with high probability temperatures (or low sample steps), \ddif~produces significantly more diverse results when the step size is minimized. Our findings demonstrate that \ddif~can generate sequences with a $60\%$ difference in each sample, whereas \textsc{PiFold} and \textsc{ProteinMPNN} achieve diversity rates below $50\%$.

We next explore the foldability of these highly diverse protein sequences designed by \ddif. Figure~\ref{fig:diff_step} compares the generated proteins (folded by \textsc{AlphaFold2}) with the crystal structure (PDB ID: 3FKF). We vary the step size over a range of ${1, 5, 10, 50}$, generating $10$ sequences per step size to calculate the average pLDDT, TM score, and diversity. For simplicity, we visualize the first structure in the figure. Decreasing the step size results in the generation of more diverse sequences. Consequently, there is a slight reduction in both pLDDT and TM scores. However, they consistently remain at a considerably high level, with both metrics approaching $0.8$. This reduction can, in part, be attributed to the increased diversity of the sequences, as \textsc{AlphaFold2} heavily relies on the MSA sequences. It is expected that more dissimilar sequences would produce a more diverse MSA. Remarkably, when step$=1$, the sequence diversity exceeds $0.6$, indicating that the generated sequences share an approximately $0.3$ sequence similarity compared to the wild-type template protein sequence. This suggests the generation of protein sequences from a substantially distinct protein family when both pLDDT and TM scores continue to exhibit a high degree of confidence.

\section{Related Work}
\paragraph{Deep Learning models for protein sequence design} 
Self-supervised models have emerged as a pivotal tool in the field of computational biology, providing a robust method for training extensive protein sequences for representation learning. These models are typically divided into two categories: structure-based generative models and sequence-based generative models.
The former approaches protein design by formulating the problem of fixed-backbone protein design as a conditional sequence generation problem. They predict node labels, which represent AA types, with invariant or equivariant graph neural networks  \cite{anand2022protein,hsu2022esmif1,ingraham2019generative,jing2021gvp,strokach2020fast,tan2023multi}. Alternatively, the latter sequence-based generative models draw parallels between protein sequences and natural language processing. They employ attention-based methods to infer residue-wise relationships within the protein structure. These methods typically recover protein sequences autoregressively conditioned on the last inferred AA \cite{madani2023progen,notin2022tranception,shin2021protein}, or employing a BERT-style generative framework with masked language modeling objectives and enable the model to predict missing or masked parts of the protein sequence \cite{lin2023esm2,meier2021esm1v,rives2021esm1b,vig2021bertology}.

\paragraph{Denoising Diffusion models}
The Diffusion Generative Model, initially introduced by Sohl-Dickstein \textit{et al.} \cite{sohl2015deep} and further developed by Ho \textit{et al.} \cite{ho2020denoising}, has emerged as a potent instrument for a myriad of generative tasks in continuous time spaces. Its applications span diverse domains, from image synthesis \cite{rombach2022high} to audio generation \cite{yang2023diffsound}, and it has also found utility in the creation of high-quality animations \cite{ho2022video}, the generation of realistic 3D objects \cite{luo2021diffusion}, and drug design \cite{corso2023diffdock,trippe2023diffusion}.
Discrete adaptations of the diffusion model, on the other hand, have demonstrated efficacy in a variety of contexts, including but not limited to, text generation \cite{austin2021structured}, image segmentation \cite{hoogeboom2021argmax}, and graph generation \cite{hoogeboom2022equivariant,vignac2023digress}. Two distinct strategies have been proposed to establish a discrete variable diffusion process. The first approach involves the transformation of categorical data into a continuous space and then applying Gaussian diffusion \cite{chen2023analog,hoogeboom2022equivariant}. The alternative strategy is to define the diffusion process directly on the categorical data, an approach notably utilized in developing the D3PM model for text generation \cite{austin2021structured}. D3PM has been further extended to graph generation, facilitating the joint generation of node features and graph structure \cite{vignac2023digress}.

\section{Conclusion}
Deep learning approaches have striven to address a multitude of critical issues in bioengineering, such as protein folding, rigid-body docking, and property prediction. However, only a few methods have successfully generated diverse sequences for fixed backbones. In this study, we offered a viable solution by developing a denoising diffusion model to generate plausible protein sequences for a predetermined backbone structure. Our method, referred to as \ddif, leverages substitution matrices for both diffusion and sampling processes, thereby exploring a practical search space for defining proteins. The iterative denoising process is predicated on the protein backbone revealing both the secondary and tertiary structure. The 3D geometry is analyzed by a modified equivariant graph neural network, which applies roto-translation equivariance to protein graphs without the necessity for intensive data augmentation. Given a protein backbone, our method successfully generated a diverse set of protein sequences, demonstrating a significant recovery rate. Importantly, these newly generated sequences are generally biologically meaningful, preserving more natural designs in the protein's conserved regions and demonstrating a high likelihood of folding back into a structure highly similar to the native protein. The design of novel proteins with desired structural and functional characteristics is of paramount importance in the biotechnology and pharmaceutical industries, where such proteins can serve diverse purposes, ranging from targeted drug delivery to enzyme design for industrial applications. Additionally, understanding how varied sequences can yield identical structures propels the exploration of protein folding principles, thereby helping to decipher the rules that govern protein folding and misfolding. Furthermore, resolving the inverse folding problem allows the identification of different sequences that fold into the same structure, shedding light on the evolutionary history of proteins by enhancing our understanding of how proteins have evolved and diversified over time while preserving their functions.

\newpage

\begin{ack}
We thank Lirong Zheng for providing insightful discussions on molecular biology. 
Bingxin Zhou acknowledges support from the National Natural Science Foundation of China (62302291).
Yu Guang Wang acknowledges support from
the National Natural Science Foundation of China (62172370), Shanghai Key Projects  (23440790200) and (2021SHZDZX0102). 
This project was undertaken with the assistance of computational resources and services from the National Computational Infrastructure (NCI), which is supported by the Australian Government.
\end{ack}

\bibliographystyle{plain}
\bibliography{reference}

\begin{thebibliography}{10}

\bibitem{alford2017rosetta}
Rebecca~F Alford, Andrew Leaver-Fay, Jeliazko~R Jeliazkov, Matthew~J O’Meara, Frank~P DiMaio, Hahnbeom Park, Maxim~V Shapovalov, P~Douglas Renfrew, Vikram~K Mulligan, Kalli Kappel, et~al.
\newblock The {Rosetta} all-atom energy function for macromolecular modeling and design.
\newblock {\em Journal of chemical theory and computation}, 13(6):3031--3048, 2017.

\bibitem{anand2022protein}
Namrata Anand and Tudor Achim.
\newblock Protein structure and sequence generation with equivariant denoising diffusion probabilistic models.
\newblock {\em arXiv preprint arXiv:2205.15019}, 2022.

\bibitem{austin2021structured}
Jacob Austin, Daniel~D Johnson, Jonathan Ho, Daniel Tarlow, and Rianne van~den Berg.
\newblock Structured denoising diffusion models in discrete state-spaces.
\newblock {\em Advances in Neural Information Processing Systems}, 34:17981--17993, 2021.

\bibitem{chen2023analog}
Ting Chen, Ruixiang ZHANG, and Geoffrey Hinton.
\newblock Analog bits: Generating discrete data using diffusion models with self-conditioning.
\newblock In {\em The Eleventh International Conference on Learning Representations}, 2023.

\bibitem{corso2023diffdock}
Gabriele Corso, Hannes St{\"a}rk, Bowen Jing, Regina Barzilay, and Tommi~S. Jaakkola.
\newblock Diffdock: Diffusion steps, twists, and turns for molecular docking.
\newblock In {\em The Eleventh International Conference on Learning Representations}, 2023.

\bibitem{dauparas2022robust}
Justas Dauparas, Ivan Anishchenko, Nathaniel Bennett, Hua Bai, Robert~J Ragotte, Lukas~F Milles, Basile~IM Wicky, Alexis Courbet, Rob~J de~Haas, Neville Bethel, et~al.
\newblock Robust deep learning--based protein sequence design using proteinmpnn.
\newblock {\em Science}, 378(6615):49--56, 2022.

\bibitem{ahmed2021prottrans}
Ahmed Elnaggar, Michael Heinzinger, Christian Dallago, Ghalia Rehawi, Wang Yu, Llion Jones, Tom Gibbs, Tamas Feher, Christoph Angerer, Martin Steinegger, Debsindhu Bhowmik, and Burkhard Rost.
\newblock {ProtTrans}: Towards cracking the language of lifes code through self-supervised deep learning and high performance computing.
\newblock {\em IEEE Transactions on Pattern Analysis and Machine Intelligence}, 2021.

\bibitem{ganea2021independent}
Octavian-Eugen Ganea, Xinyuan Huang, Charlotte Bunne, Yatao Bian, Regina Barzilay, Tommi~S Jaakkola, and Andreas Krause.
\newblock Independent {SE}(3)-equivariant models for end-to-end rigid protein docking.
\newblock In {\em International Conference on Learning Representations}, 2021.

\bibitem{gao2022alphadesign}
Zhangyang Gao, Cheng Tan, and Stan~Z Li.
\newblock Alphadesign: A graph protein design method and benchmark on alphafolddb.
\newblock {\em arXiv preprint arXiv:2202.01079}, 2022.

\bibitem{gao2023pifold}
Zhangyang Gao, Cheng Tan, and Stan~Z. Li.
\newblock Pifold: Toward effective and efficient protein inverse folding.
\newblock In {\em International Conference on Learning Representations}, 2023.

\bibitem{gong2017improving}
Hai’e Gong, Haicang Zhang, Jianwei Zhu, Chao Wang, Shiwei Sun, Wei-Mou Zheng, and Dongbo Bu.
\newblock Improving prediction of burial state of residues by exploiting correlation among residues.
\newblock {\em BMC bioinformatics}, 18(3):165--175, 2017.

\bibitem{henikoff1992blosum}
Steven Henikoff and Jorja~G Henikoff.
\newblock Amino acid substitution matrices from protein blocks.
\newblock {\em Proceedings of the National Academy of Sciences}, 89(22):10915--10919, 1992.

\bibitem{ho2020denoising}
Jonathan Ho, Ajay Jain, and Pieter Abbeel.
\newblock Denoising diffusion probabilistic models.
\newblock {\em Advances in Neural Information Processing Systems}, 33:6840--6851, 2020.

\bibitem{ho2022video}
Jonathan Ho, Tim Salimans, Alexey~A. Gritsenko, William Chan, Mohammad Norouzi, and David~J. Fleet.
\newblock Video diffusion models.
\newblock In Alice~H. Oh, Alekh Agarwal, Danielle Belgrave, and Kyunghyun Cho, editors, {\em Advances in Neural Information Processing Systems}, 2022.

\bibitem{hoogeboom2021argmax}
Emiel Hoogeboom, Didrik Nielsen, Priyank Jaini, Patrick Forr{\'e}, and Max Welling.
\newblock Argmax flows and multinomial diffusion: Learning categorical distributions.
\newblock {\em Advances in Neural Information Processing Systems}, 34:12454--12465, 2021.

\bibitem{hoogeboom2022equivariant}
Emiel Hoogeboom, V{\i}ctor~Garcia Satorras, Cl{\'e}ment Vignac, and Max Welling.
\newblock Equivariant diffusion for molecule generation in 3d.
\newblock In {\em International Conference on Machine Learning}, pages 8867--8887. PMLR, 2022.

\bibitem{hsu2022esmif1}
Chloe Hsu, Robert Verkuil, Jason Liu, Zeming Lin, Brian Hie, Tom Sercu, Adam Lerer, and Alexander Rives.
\newblock Learning inverse folding from millions of predicted structures.
\newblock In Kamalika Chaudhuri, Stefanie Jegelka, Le~Song, Csaba Szepesvari, Gang Niu, and Sivan Sabato, editors, {\em Proceedings of the 39th International Conference on Machine Learning}, volume 162 of {\em Proceedings of Machine Learning Research}, pages 8946--8970. PMLR, 17--23 Jul 2022.

\bibitem{ingraham2019generative}
John Ingraham, Vikas Garg, Regina Barzilay, and Tommi Jaakkola.
\newblock Generative models for graph-based protein design.
\newblock {\em Advances in Neural Information Processing Systems}, 32, 2019.

\bibitem{jing2022torsional}
Bowen Jing, Gabriele Corso, Jeffrey Chang, Regina Barzilay, and Tommi~S. Jaakkola.
\newblock Torsional diffusion for molecular conformer generation.
\newblock In Alice~H. Oh, Alekh Agarwal, Danielle Belgrave, and Kyunghyun Cho, editors, {\em Advances in Neural Information Processing Systems}, 2022.

\bibitem{jing2021gvp}
Bowen Jing, Stephan Eismann, Patricia Suriana, Raphael John~Lamarre Townshend, and Ron Dror.
\newblock Learning from protein structure with geometric vector perceptrons.
\newblock In {\em International Conference on Learning Representations}, 2021.

\bibitem{jumper2021alphafold}
John Jumper, Richard Evans, Alexander Pritzel, Tim Green, Michael Figurnov, Olaf Ronneberger, Kathryn Tunyasuvunakool, Russ Bates, Augustin {\v{Z}}{\'\i}dek, Anna Potapenko, et~al.
\newblock Highly accurate protein structure prediction with {AlphaFold}.
\newblock {\em Nature}, 596(7873):583--589, 2021.

\bibitem{khoury2014protein}
George~A Khoury, James Smadbeck, Chris~A Kieslich, and Christodoulos~A Floudas.
\newblock Protein folding and de novo protein design for biotechnological applications.
\newblock {\em Trends in biotechnology}, 32(2):99--109, 2014.

\bibitem{lin2023esm2}
Zeming Lin, Halil Akin, Roshan Rao, Brian Hie, Zhongkai Zhu, Wenting Lu, Nikita Smetanin, Robert Verkuil, Ori Kabeli, Yaniv Shmueli, et~al.
\newblock Evolutionary-scale prediction of atomic-level protein structure with a language model.
\newblock {\em Science}, 379(6637):1123--1130, 2023.

\bibitem{luo2021diffusion}
Shitong Luo and Wei Hu.
\newblock Diffusion probabilistic models for 3d point cloud generation.
\newblock In {\em Proceedings of the IEEE/CVF Conference on Computer Vision and Pattern Recognition}, pages 2837--2845, 2021.

\bibitem{madani2023progen}
Ali Madani, Ben Krause, Eric~R Greene, Subu Subramanian, Benjamin~P Mohr, James~M Holton, Jose~Luis Olmos~Jr, Caiming Xiong, Zachary~Z Sun, Richard Socher, et~al.
\newblock Large language models generate functional protein sequences across diverse families.
\newblock {\em Nature Biotechnology}, pages 1--8, 2023.

\bibitem{meier2021esm1v}
Joshua Meier, Roshan Rao, Robert Verkuil, Jason Liu, Tom Sercu, and Alex Rives.
\newblock Language models enable zero-shot prediction of the effects of mutations on protein function.
\newblock In {\em Advances in Neural Information Processing Systems}, volume~34, pages 29287--29303, 2021.

\bibitem{nichol2021improved}
Alexander~Quinn Nichol and Prafulla Dhariwal.
\newblock Improved denoising diffusion probabilistic models.
\newblock In {\em International Conference on Machine Learning}, pages 8162--8171. PMLR, 2021.

\bibitem{nijkamp2022progen2}
Erik Nijkamp, Jeffrey Ruffolo, Eli~N Weinstein, Nikhil Naik, and Ali Madani.
\newblock {ProGen2}: exploring the boundaries of protein language models.
\newblock {\em arXiv:2206.13517}, 2022.

\bibitem{notin2022tranception}
Pascal Notin, Mafalda Dias, Jonathan Frazer, Javier~Marchena Hurtado, Aidan~N Gomez, Debora Marks, and Yarin Gal.
\newblock Tranception: protein fitness prediction with autoregressive transformers and inference-time retrieval.
\newblock In {\em International Conference on Machine Learning}, pages 16990--17017. PMLR, 2022.

\bibitem{ORENGO1997cath}
CA~Orengo, AD~Michie, S~Jones, DT~Jones, MB~Swindells, and JM~Thornton.
\newblock {CATH} – a hierarchic classification of protein domain structures.
\newblock {\em Structure}, 5(8):1093--1109, 1997.

\bibitem{pearce2021deep}
Robin Pearce and Yang Zhang.
\newblock Deep learning techniques have significantly impacted protein structure prediction and protein design.
\newblock {\em Current opinion in structural biology}, 68:194--207, 2021.

\bibitem{qi2020densecpd}
Yifei Qi and John~ZH Zhang.
\newblock Densecpd: improving the accuracy of neural-network-based computational protein sequence design with densenet.
\newblock {\em Journal of chemical information and modeling}, 60(3):1245--1252, 2020.

\bibitem{rao2021msa}
Roshan~M Rao, Jason Liu, Robert Verkuil, Joshua Meier, John Canny, Pieter Abbeel, Tom Sercu, and Alexander Rives.
\newblock {MSA} transformer.
\newblock In {\em International Conference on Machine Learning}, pages 8844--8856. PMLR, 2021.

\bibitem{rives2021esm1b}
Alexander Rives, Joshua Meier, Tom Sercu, Siddharth Goyal, Zeming Lin, Jason Liu, Demi Guo, Myle Ott, C~Lawrence Zitnick, Jerry Ma, et~al.
\newblock Biological structure and function emerge from scaling unsupervised learning to 250 million protein sequences.
\newblock {\em Proceedings of the National Academy of Sciences}, 118(15):e2016239118, 2021.

\bibitem{rombach2022high}
Robin Rombach, Andreas Blattmann, Dominik Lorenz, Patrick Esser, and Bj{\"o}rn Ommer.
\newblock High-resolution image synthesis with latent diffusion models.
\newblock In {\em Proceedings of the IEEE/CVF Conference on Computer Vision and Pattern Recognition}, pages 10684--10695, 2022.

\bibitem{satorras2021n}
V{\i}ctor~Garcia Satorras, Emiel Hoogeboom, and Max Welling.
\newblock E(n) equivariant graph neural networks.
\newblock In {\em International conference on machine learning}, pages 9323--9332, 2021.

\bibitem{shin2021protein}
Jung-Eun Shin, Adam~J Riesselman, Aaron~W Kollasch, Conor McMahon, Elana Simon, Chris Sander, Aashish Manglik, Andrew~C Kruse, and Debora~S Marks.
\newblock Protein design and variant prediction using autoregressive generative models.
\newblock {\em Nature Communications}, 12(1):2403, 2021.

\bibitem{sledz2018protein}
Pawe{\l} {\'S}led{\'z} and Amedeo Caflisch.
\newblock Protein structure-based drug design: from docking to molecular dynamics.
\newblock {\em Current opinion in structural biology}, 48:93--102, 2018.

\bibitem{sohl2015deep}
Jascha Sohl-Dickstein, Eric Weiss, Niru Maheswaranathan, and Surya Ganguli.
\newblock Deep unsupervised learning using nonequilibrium thermodynamics.
\newblock In {\em International Conference on Machine Learning}, pages 2256--2265. PMLR, 2015.

\bibitem{songdenoising}
Jiaming Song, Chenlin Meng, and Stefano Ermon.
\newblock Denoising diffusion implicit models.
\newblock In {\em International Conference on Learning Representations}.

\bibitem{strokach2020fast}
Alexey Strokach, David Becerra, Carles Corbi-Verge, Albert Perez-Riba, and Philip~M Kim.
\newblock Fast and flexible protein design using deep graph neural networks.
\newblock {\em Cell Systems}, 11(4):402--411, 2020.

\bibitem{tan2022generative}
Cheng Tan, Zhangyang Gao, Jun Xia, and Stan~Z Li.
\newblock Generative de novo protein design with global context.
\newblock {\em arXiv preprint arXiv:2204.10673}, 2022.

\bibitem{tan2023multi}
Yang Tan, Bingxin Zhou, Yuanhong Jiang, Yu~Guang Wang, and Liang Hong.
\newblock Multi-level protein representation learning for blind mutational effect prediction.
\newblock {\em arXiv:2306.04899}, 2023.

\bibitem{trippe2023diffusion}
Brian~L. Trippe, Jason Yim, Doug Tischer, David Baker, Tamara Broderick, Regina Barzilay, and Tommi~S. Jaakkola.
\newblock Diffusion probabilistic modeling of protein backbones in 3d for the motif-scaffolding problem.
\newblock In {\em International Conference on Learning Representations}, 2023.

\bibitem{trivedi2020substitution}
Rakesh Trivedi and Hampapathalu~Adimurthy Nagarajaram.
\newblock Substitution scoring matrices for proteins-an overview.
\newblock {\em Protein Science}, 29(11):2150--2163, 2020.

\bibitem{van2008visualizing}
Laurens Van~der Maaten and Geoffrey Hinton.
\newblock Visualizing data using t-sne.
\newblock {\em Journal of machine learning research}, 9(11), 2008.

\bibitem{vig2021bertology}
Jesse Vig, Ali Madani, Lav~R Varshney, Caiming Xiong, Nazneen Rajani, et~al.
\newblock {BERTology} meets biology: Interpreting attention in protein language models.
\newblock In {\em International Conference on Learning Representations}, 2021.

\bibitem{vignac2023digress}
Clement Vignac, Igor Krawczuk, Antoine Siraudin, Bohan Wang, Volkan Cevher, and Pascal Frossard.
\newblock Digress: Discrete denoising diffusion for graph generation.
\newblock In {\em The Eleventh International Conference on Learning Representations}, 2023.

\bibitem{wu2022protein}
Kevin~E Wu, Kevin~K Yang, Rianne van~den Berg, James~Y Zou, Alex~X Lu, and Ava~P Amini.
\newblock Protein structure generation via folding diffusion.
\newblock {\em arXiv preprint arXiv:2209.15611}, 2022.

\bibitem{yang2023diffsound}
Dongchao Yang, Jianwei Yu, Helin Wang, Wen Wang, Chao Weng, Yuexian Zou, and Dong Yu.
\newblock Diffsound: Discrete diffusion model for text-to-sound generation.
\newblock {\em IEEE/ACM Transactions on Audio, Speech, and Language Processing}, 2023.

\bibitem{zhou2023accurate}
Bingxin Zhou, Outongyi Lv, Kai Yi, Xinye Xiong, Pan Tan, Liang Hong, and Yu~Guang Wang.
\newblock Accurate and definite mutational effect prediction with lightweight equivariant graph neural networks.
\newblock {\em arXiv:2304.08299}, 2023.

\bibitem{zhou2023conditional}
Bingxin Zhou, Lirong Zheng, Banghao Wu, Kai Yi, Bozitao Zhong, Pietro Lio, and Liang Hong.
\newblock Conditional protein denoising diffusion generates programmable endonucleases.
\newblock {\em bioRxiv}, pages 2023--08, 2023.

\end{thebibliography}

%%%%%%%%%%%%%%%%%%%%%%%%%%%%%%%%%%%%%%%%%%%%%%%%%%%%%%%%%%%%
\newpage
\appendix
\section{Broader Impact and Limitations}

\paragraph{Broader Impact}
We have developed a generative model rooted in the diffusion denoising paradigm, specifically tailored to the context of protein inverse folding. As with any other generative models, it is capable of generating \textit{de novo} content (protein sequences) under specified conditions (e.g., protein tertiary structure). While this method holds substantial potential for facilitating scientific research and biological discoveries, its misuse could pose potential risks to human society. For instance, in theory, it possesses the capacity to generate novel viral protein sequences with enhanced functionalities. To mitigate this potential risk, one approach could be to confine the training dataset for the model to proteins derived from prokaryotes and/or eukaryotes, thereby excluding viral proteins. Although this strategy may to some extent compromise the overall performance and generalizability of the trained model, it also curtails the risk of misuse of the model by limiting the understanding and analysis of viral protein construction.

\paragraph{Limitations}
The conditions imposed on the sampling process gently guide the generated protein sequences. However, in certain scenarios, stringent restrictions may be necessary to produce a functional protein. Secondary structure, as a living example, actively contributes to the protein's functionality. For instance, 
transmembrane $\alpha$-helices play essential roles in protein functions, such as passing ions or other molecules and transmitting a signal across the membrane. Moreover, the current zero-shot model is trained on a general protein database. For specific downstream applications, such as generating new sequences for a particular protein or protein family, it may necessitate the incorporation of auxiliary modules or the modification of training procedures to yield more fitting sequences.

\section{Non-Markovian Forward Process}
We give the derivation of posterior distribution $q\left(\vx_{t-1} \mid  \vx_t,\vx^{\rm aa}\right)$ for generative process from step to step. The proof relies on the Bayes rule, Markov property, and the pre-defined transition matrix for AAs.
\begin{prop}
\label{prop1}
For $q\left(\vx_{t-1} \mid  \vx_t,\vx^{\rm aa}\right)$ defined in Eq \ref{eq:posteriorDist}, we have 
\begin{equation*}
   q\left(\vx_{t-1} \mid  \vx_t,\vx^{\rm aa}\right) = \text{\rm Cat}\left(\vx_{t-1}\Big|\frac{\vx_tQ_t^{\top} \odot \vx^{\rm aa}\bar{Q}^{t-1}}{\vx^{\rm aa}\bar{Q}_t\vx_t^{\top}}\right).
\end{equation*}
\end{prop}

\begin{proof}
By Bayes rules, we can expand the original equation $q\left(\vx_{t-1} \mid  \vx_t,\vx^{\rm aa}\right)$ to 
\begin{equation*}
    q\left(\vx_{t-1} \mid \vx_t, \vx^{\rm aa}\right)=\frac{q\left(\vx_t \mid \vx_{t-1}, \vx^{\rm aa}\right) q\left(\vx_{t-1} \mid \vx^{\rm aa}\right)}{q\left(\vx_t \mid \vx^{\rm aa}\right)}=\frac{q\left(\vx_t \mid \vx_{t-1} \right) q\left(\vx_{t-1} \mid \vx^{\rm aa}\right)}{q\left(\vx_t \mid \vx^{\rm aa}\right)}.
\end{equation*}

As pre-defined diffusion process, we get $q\left(\vx_t \mid \vx^{\rm aa}\right)=\vx^{\rm aa}\bar{Q}_{t}$, and $q\left(\vx_{t-1} \mid \vx^{\rm aa}\right)=\vx^{\rm aa}\bar{Q}_{t-1}$.

For the term of $q\left(\vx_t \mid \vx_{t-1}, \vx^{\rm aa}\right)$ by Bayes rule and Markov property, we have 
$$q\left(\vx_t \mid \vx_{t-1}, \vx^{\rm aa}\right)=q\left(\vx_t \mid \vx_{t-1}\right) \propto q(\vx_{t-1} \mid \vx_t)\pi(\vx_{t}) \propto \vx_{t}Q_t^\top \odot\pi(\vx_{t})$$
where the normalizing constant is $\sum_{\vx_{t-1}}\vx_{t}Q_t^\top \odot \pi(\vx_{t})=(\vx_{t}\sum_{\vx_{t-1}} Q_t^{\top}) \odot \pi(\vx_{t})=\vx_{t} \odot\pi(\vx_{t})$

Then $q\left(\vx_t \mid \vx_{t-1}, \vx^{\rm aa}\right) = \frac{\vx_{t}Q_t^\top }{\vx_{t}}$, and the posterior distribution is:
\begin{equation*}
    q\left(\vx_{t-1} \mid \vx^{\rm aa}, \vx_t\right) = \text{\rm Cat}\left(\vx_{t-1}\Big|\frac{\vx_tQ^{\top}_t \odot \vx^{\rm aa}\bar{Q}_{t-1}}{\vx^{\rm aa}\bar{Q}_t\vx^{\top}_t}\right).
\end{equation*}
\end{proof}

The following gives the derivation for the discrete DDIM which accelerates the generative process.
\begin{prop}
For $q\left(\vx_{t-k} \mid  \vx_t,\vx^{\rm aa}\right)$ defined in Eq \ref{eq:ddimPosterior},   
\begin{equation*}
   q\left(\vx_{t-k} \mid  \vx_t,\vx^{\rm aa}\right) = \text{\rm Cat}\left(\vx_{t-k}\Big|\frac{\vx_t Q_t^{\top}\cdots Q_{t-k}^{\top } \odot \vx^{\rm aa}\bar{Q}_{t-k}}{\vx^{\rm aa}\bar{Q}_t\vx_t^{\top}}\right). 
\end{equation*}

\begin{proof}
By Bayes rules, we can expand the original equation $q\left(\vx_{t-k} \mid  \vx_t,\vx^{\rm aa}\right)$ to 
\begin{equation*}
    q\left(\vx_{t-k} \mid \vx_t, \vx^{\rm aa}\right)=\frac{q\left(\vx_t \mid \vx_{t-k}, \vx^{\rm aa}\right) q\left(\vx_{t-k} \mid \vx^{\rm aa}\right)}{q\left(\vx_t \mid \vx^{\rm aa}\right)}=\frac{q\left(\vx_t \mid \vx_{t-k} \right) q\left(\vx_{t-k} \mid \vx^{\rm aa}\right)}{q\left(\vx_t \mid \vx^{\rm aa}\right)}.
\end{equation*}

As pre-defined diffusion process, we get $q\left(\vx_t \mid \vx^{\rm aa}\right)=\vx^{\rm aa}\bar{Q}_{t}$, and $q\left(\vx_{t-1} \mid \vx^{\rm aa}\right)=\vx^{\rm aa}\bar{Q}_{t-k}$.

Similarly with $q\left(\vx_t \mid \vx_{t-1}, \vx^{\rm aa}\right)$ in Proposition~\ref{prop1},  $q\left(\vx_t \mid \vx_{t-k}, \vx^{\rm aa}\right) = \frac{\vx_{t}Q_t^{\top}\cdots Q_{t-k}^{\top } }{\vx_{t}}$ and the posterior is 
\begin{equation*}
    q\left(\vx_{t-k} \mid  \vx_t,\vx^{\rm aa}\right) = \text{\rm Cat}\left(\vx_{t-k}\Big|\frac{\vx_t Q_t^{\top}\cdots Q_{t-k}^{\top } \odot \vx^{\rm aa}\bar{Q}_{t-k}}{\vx^{\rm aa}\bar{Q}_t\vx_t^{\top}}\right).
\end{equation*}
\end{proof}
\end{prop}

\section{Graph Representation of Folded Proteins}
\label{sec:app:proteinGraph}
The geometry of proteins suggests higher-level structures and topological relationships, which are vital to protein functionality. For a given protein, we create a $k$-nearest neighbor ($k$NN) graph $\gG=(\mX,\mE)$ to describe its physiochemical and geometric properties with nodes representing AAs by $\mX\in\mathbb{R}^{39}$ node attributes with $20$-dim AA type encoder, $16$-dim AA properties, and $3$-dim AA positions. The undirected edge connections are formulated via a $k$NN-graph with cutoff. In other words, each node is connected to up to $k$ other nodes in the graph that has the smallest Euclidean distance over other nodes and the distance is smaller than a certain cutoff (\eg $30$\AA). Edge attributes are defined for connected node pairs. For instance, if node $i$ and $j$ are connected to each other, their relationship will be described by $\mE_{ij}=\mE_{ji}\in\mathbb{R}^{93}$.

The AA types are one-hot encoded to $20$ binary values by $\mX^{\rm aa}$. On top of it, the properties of AAs and AAs' local environment are described by $\mX^{\rm prop}$, including the normalized crystallographic B-factor, solvent-accessible surface area (SASA), normalized surface-aware node features, dihedral angles of backbone atoms, and 3D positions. SASA measures the level of exposure of an AA to solvent in a protein by a scalar value, which provides an important indicator of active sites of proteins to locate whether a residue is on the surface of the protein. Both B-factor and SASA are standardized  with AA-wise mean and standard deviation on the associate attribute. Surface-aware features \cite{ganea2021independent} of an AA is non-linear projections to the weighted average distance of the central AA to its one-hop neighbors $i^{\prime} \in \mathcal{N}_i$, \ie 
\begin{equation*}
      \rho\left(\mathbf{x}_i ; \lambda\right)=\frac{\left\|\sum_{i^{\prime} \in \mathcal{N}_i} w_{i, i^{\prime}, \lambda}\left(\mX^{\text{pos},i}-\mX^{\text{pos},i^{\prime}}\right)\right\|}{\sum_{i^{\prime} \in \mathcal{N}_i} w_{i, i^{\prime}, \lambda}\left\|\mX^{\text{pos},i}-\mX^{\text{pos},i^{\prime}}\right\|}, 
\end{equation*}
where the weights are defined by
\begin{equation*}
    w_{i, i^{\prime}, \lambda}=\frac{\exp \left(-\left\|\mX_{\text{pos},i}-\mX_{\text{pos},i^{\prime}}\right\|^2 / \lambda\right)}{\sum_{i^{\prime} \in \mathcal{N}_i} \exp \left(-\left\|\mX_{\text{pos},i}-\mX_{\text{pos},i^{\prime}}\right\|^2 / \lambda\right)}
\end{equation*}
with $\lambda\in\{1,2,5,10,30\}$. The $\mX^{\text{pos},i}\in\mathbb{R}^3$ denotes the \emph{3D coordinates} of the $i$th residue, which is represented by the position of $\alpha$-carbon. We also use the backbone atom positions to define the spatial conformation of each AA in the protein chain with trigonometric values of dihedral angles $\{\sin, \cos\} \circ \{\phi_i, \psi_i,\omega_i\}$. %Note that the last node in the protein sequence is removed to avoid inaccessible angles. 

Edge attributes $\mE\in\mathbb{R}^{93}$, on the other hand, include kernel-based distances, relative spatial positions, and relative sequential distances for pairwise distance characterization. For two connected residues $i$ and $j$, the kernel-based distance between them is projected by Gaussian radial basis functions (RBF) of $\exp\left\{\frac{\|\vx_j-\vx_i\|^2}{2\sigma^2_r}\right\}$ with $r=1,2,\dots,R$. 
% \begin{equation*}
%       \mE_{r}^{\rm rbf}(\vx_i,\vx_j)= \exp\left\{\frac{\|\vx_j-\vx_i\|^2}{2\sigma^2_r}\right\}, \quad r=1,2,\dots,R.
% \end{equation*}
A total number of $15$ distinct distance-based features are created with $\sigma_r=\{1.5^k\mid k=0,1,2,\dots,14\}$. Next, local frames \cite{ganea2021independent} are created from the corresponding residues' heavy atoms positions to define $12$ relative positions. They represent local fine-grained relations between AAs and the rigid property of how the two residues interact with each other. Finally, the residues' sequential relationship is encoded with $66$ binary features by their relative position $d_{i,j} = \vert s_i - s_j\vert$, where $s_i$ and $s_j$ are the absolute positions of the two nodes in the AA chain \cite{zhou2023accurate}. We further define a binary contact signal \cite{ingraham2019generative} to indicate whether two residues contact in the space, \ie the Euclidean distance $\|C\alpha_i-C\alpha_j\|<8$.

\section{Training and Inference}
\label{sec:app:algorithm}
In this section, we elucidate the training and inference methodologies implemented in the diffusion generative model. As shown in Algorithm \ref{algorith:Train}, training commences with a random sampling of a time scale $t$ from a uniform distribution between $1$ and $T$. Subsequently, we calculate the noise posterior and integrate noise as dictated by its respective distribution. We then utilize an equivariant graph neural network for denoising predictions, using both the noisy amino acid and other properties as node features, and leveraging the graph structure for geometric information. This results in the model outputting the denoised amino acid type. Ultimately, the cross-entropy loss is computed between the predicted and original amino acid types, providing a parameter for optimizing the neural network.

\begin{algorithm}[H]
\caption{Training}  
\label{algorith:Train}
\begin{algorithmic}[1]
    \State \textbf{Input}: A graph $\gG = \{\mX,\mE\}$
    \State Sample $t \sim \mathcal{U}(1,T)$
    \State Compute $q(\mX_t | \mX^{\rm aa}) =\mX^{\rm aa}\bar{Q}_t$
    \State Sample noisy $\mX_t \sim q(\mX_t|\mX^{\rm aa})$
    \State Forward pass: $\hat{p}(\mX^{\rm aa}) =f_{\theta}(\mX_t,\mE,t,ss)$ 
    \State Compute cross-entropy loss: $L = L_{\text{CE}}(\hat{p}(\mX^{\rm aa}),\mX)$
    \State Compute the gradient and optimize denoise network $f_{\theta}$
\end{algorithmic}  
\end{algorithm}

Upon completing the training, we are capable of sampling data using the neural network and the posterior distribution $p(\vx_{t-1}|\vx_t,\vx^{\rm aa})$. As delineated in the algorithm, we initially sample an amino acid uniformly from 20 classes, then employ our neural network to denoise $\mX^{\rm aa}$ from time $t$. From here, we can calculate the forward probability utilizing the model output and the posterior distribution. Through iterative processing, the ultimate model sample closely approximates the original data distribution. More importantly, we illustrate how to speed up the sampling procedure using DDIM in Algorithm~\ref{algorithm:ddim}. It can be regarded as skipping several steps in DDPM but with close performance (see Figure~\ref{fig:ddim} in Section~\ref{sec:inversen folding}). DDPM is a special case of DDIM when skipping step $k=1$.
We can also partially generate amino acids given some known AA within the structure. Analogous to an inpainting task, at each step, we can manually adjust the prediction at known positions to the known amino acid type, subsequently introducing noise, as illustrated in the Algorithm \ref{algorithm:inpainting}.

% \bxc{need to specify why DDPM is shown here (\eg for comparison). Also maybe you can highlight the different part between Algorithm 2 and 3.}
% \begin{minipage}{\textwidth}
% \centering
% \begin{minipage}{0.48\textwidth}
%     \centering
%     \begin{algorithm}[H]
%         \caption{Training}  
%         \begin{algorithmic}[1]
%       \State \textbf{Input}: A graph $G = \{X,E\}$
%       \State Sample $t \sim \mathcal{U}(1,T)$
%       \State Compute $q(x_t | x_0) = x_0\bar{Q}_t$
%       \State Sample noisy $x_t \sim q(x_t|x_0)$
%       \State Forward pass: $\hat{p}(x_0) =f_{\theta}(x_t,t,E,ss)$ 
%       \State Compute Loss: $l = l_{\text{CE}}(\hat{p}(x_0),X)$
%       \State Compute the gradient and optimize denoise network $f_{\theta}$
%        \end{algorithmic}  
%       \end{algorithm}
% \end{minipage}
% \begin{minipage}{0.48\textwidth}
%     \begin{algorithm}[H]
%         \caption{Sampling (DDPM)}  
%         \begin{algorithmic}[1]
%       \State Sample from prior $X_T \sim p(X_T)$ (uniform distributed on 20 class) 
%       \For {$t$ in $\{T,T-1,...,1\}$}
%       \State Predict $\hat{p}(x_0|x_t) =f_{\theta}(x_t,t,E,ss)$
%       % \State Sample $\hat{x}_0 \sim \hat{p}(x_0|x_t)$
%       \State Compute $p_{\theta}(x_{t-1}|x_t) = \sum_{\hat{x_0}}q(x_{t-1}|x_t,\hat{x}_0)\hat{p}(x_0|x_t)$
%       \State Sample $x_{t-1} \sim p_{\theta}(x_{t-1}|x_t)$
      
%       \EndFor
%       \State Sample $x \sim p_{\theta}(x_0|x_1)$
%        \end{algorithmic}  
%       \end{algorithm}
% \end{minipage}
% \end{minipage}

\begin{algorithm}[H]
\caption{Sampling (DDPM)}  
\label{algorithm:ddpm}
\begin{algorithmic}[1]
    \State Sample from uniformly prior $\mX_T \sim p(\mX_T)$
    \For {$t$ in $\{T,T-1,...,1\}$}
    \State Predict $\hat{p}(\mX^{\rm aa}|\,\mX_t)$ by neural network $\hat{p}(\mX^{\rm aa}|\,\mX_t) =f_{\theta}(\mX_t,\mE,t,ss)$
    % \State Sample $\hat{x}_0 \sim \hat{p}(x_0|x_t)$
    \State Compute $p_{\theta}(\mX_{t-1}|\mX_t) = \sum_{\hat{\mX}^{\rm aa}}q(\mX_{t-1}|\mX_t,\hat{\mX}^{\rm aa})\hat{p}(\mX^{\rm aa}|\mX_t)$
    \State Sample $\mX_{t-1} \sim p_{\theta}(\mX_{t-1}|\mX_t)$
    
    \EndFor
    \State Sample $\mX^{\rm aa} \sim p_{\theta}(\mX^{\rm aa}|\mX_1)$
\end{algorithmic}
\end{algorithm}

\begin{algorithm}[H]
\caption{Sampling (DDIM)}  
\label{algorithm:ddim}
\begin{algorithmic}[1]
    \State Sample from uniformly prior $\mX_T \sim p(\mX_T)$
    \For {$t$ in $\{T,T-k,...,1\}$}
    \State Predict $\hat{p}(\mX^{\rm aa}|\,\mX_t)$ by neural network $\hat{p}(\mX^{\rm aa}|\,\mX_t) =f_{\theta}(\mX_t,\mE,t,ss)$
    % \State Sample $\hat{x}_0 \sim \hat{p}(x_0|x_t)$
    \State Compute $p_{\theta}(\mX_{t-k}|\mX_t) = \sum_{\hat{\mX}^{\rm aa}}q(\mX_{t-k}|\mX_t,\hat{\mX}^{\rm aa})\hat{p}(\mX^{\rm aa}|\mX_t)$
    \State Sample $\mX_{t-k} \sim p_{\theta}(\mX_{t-k}|\mX_t)$
    
    \EndFor
    \State Sample $\mX^{\rm aa} \sim p_{\theta}(\mX^{\rm aa}|\mX_1)$
\end{algorithmic}
\end{algorithm}

\begin{algorithm}[H]
\caption{Partial Sampling}  
\label{algorithm:inpainting}
\begin{algorithmic}[1]
    \State Input Mask $M$ indicate which position is fixed
    \State Sample from uniformly prior $\mX_T \sim p(\mX_T)$
    \For {$t$ in $\{T,T-k,...,1\}$}
    \State Predict $\hat{p}(\mX^{\rm aa}|\,\mX_t)$ by neural network $\hat{p}(\mX^{\rm aa}|\,\mX_t) =f_{\theta}(\mX_t,\mE,t,ss)$
    \State Compute $\hat{p}(\mX^{\rm aa}|\,\mX_t) = p(\mX^{\rm aa}) \odot M + \hat{p}(\mX^{\rm aa}|\,\mX_t) \odot (1-M)$ 
    % \State Sample $\hat{x}_0 \sim \hat{p}(x_0|x_t)$
    \State Compute $p_{\theta}(\mX_{t-k}|\mX_t) = \sum_{\hat{\mX}^{\rm aa}}q(\mX_{t-k}|\mX_t,\hat{\mX}^{\rm aa})\hat{p}(\mX^{\rm aa}|\mX_t)$
    \State Sample $\mX_{t-k} \sim p_{\theta}(\mX_{t-k}|\mX_t)$
    \EndFor
    \State Sample $\mX^{\rm aa} \sim p_{\theta}(\mX^{\rm aa}|\mX_1)$
\end{algorithmic}
\end{algorithm}

\section{Inverse Folding Performance on TS50 and T500}
\label{sec:app:ts50500}
In addition to the \textbf{CATH} dataset, we also evaluated our model using the \textbf{TS50} and \textbf{T500} datasets. These datasets were introduced by DenseCPD \cite{qi2020densecpd}, encompassing 9888 structures for training, and two distinct test datasets comprising 50 (TS50) and 500 (T500) test datasets, respectively. The same preprocessing steps applied to the \textbf{CATH} dataset were utilized here. The denoising network comprises six sequentially arranged EGNN blocks, each boasting a hidden dimension of $256$. Our model's performance, outlined in Table \ref{tab:rr on ts}, achieved an accuracy of $61.22\%$ on \textbf{T500}, and $56.32\%$ on \textbf{TS50}, respectively.
\begin{table}[H]
\caption{Recovery rate performance of \textbf{TS50} and \textbf{T500} on zero-shot models.}
\label{tab:rr on ts}
\begin{center}
\resizebox{0.8\textwidth}{!}{
    \begin{tabular}{ccccc}
    \toprule \multirow{2}{*}{\textbf{Model}} & \multicolumn{2}{c}{\textbf{TS50}} & \multicolumn{2}{c}{\textbf{T500}} \\\cmidrule(lr){2-3}\cmidrule(lr){4-5}
     & Perplexity $\downarrow$ & Recovery $\uparrow$ & Perplexity $\downarrow$ & Recovery $\uparrow$ \\
    \midrule 
    \textsc{StructGNN}  \cite{ingraham2019generative} & 5.40 & 43.89 & 4.98 & 45.69 \\
    \textsc{GraphTrans} \cite{ingraham2019generative}& 5.60 & 42.20 & 5.16 & 44.66 \\
    \textsc{GVP}  \cite{jing2021gvp}& 4.71 & 44.14 & 4.20 & 49.14 \\
    \textsc{GCA} \cite{tan2022generative}& 5.09 & 47.02 & 4.72 & 47.74 \\
    \textsc{AlphaDesign} \cite{gao2022alphadesign}& 5.25 & 48.36 & 4.93 & 49.23 \\
    \textsc{ProteinMPNN}  \cite{dauparas2022robust}& 3.93 & 54.43 & 3.53 & 58.08 \\
    \textsc{PiFold} \cite{gao2023pifold}& 3.86 & \textbf{58.72} & 3.44 & 60.42 \\
    \midrule
    \ddif (ours) & \textbf{3.71} & 56.32 & \textbf{3.23} & \textbf{61.22} \\
    \bottomrule
    \end{tabular}
}
\end{center}
\end{table}

\section{Ablation Study}
\label{sec:app:ablation}
We conducted ablation studies to assess the impact of various factors on our model's performance. These elements encompassed the selection of the transition matrix (uniform versus BLOSUM), the integration of secondary structure embeddings in the denoising procedure, and the function of the equivariant neural network. As demonstrated in Figure \ref{fig:recoverRate_SubMtr}, incorporating equivariance into the denoising neural network substantially enhances the model's performance. Given that the placement of protein structures in space can be arbitrary, considering symmetry in the denoising neural network helps to mitigate disturbances. Moreover, we found that including secondary structure as auxiliary information lessens uncertainty and improves recovery. Lastly, utilizing the BLOSUM matrix as the noise transition matrix boosted the recovery rate by 2\%, highlighting the benefits of infusing biological information into the diffusion and generative processes. This approach reduces sample variance and substantially benefits overall model performance.

\begin{figure}[!ht]
    \centering
    \includegraphics[width=\textwidth]{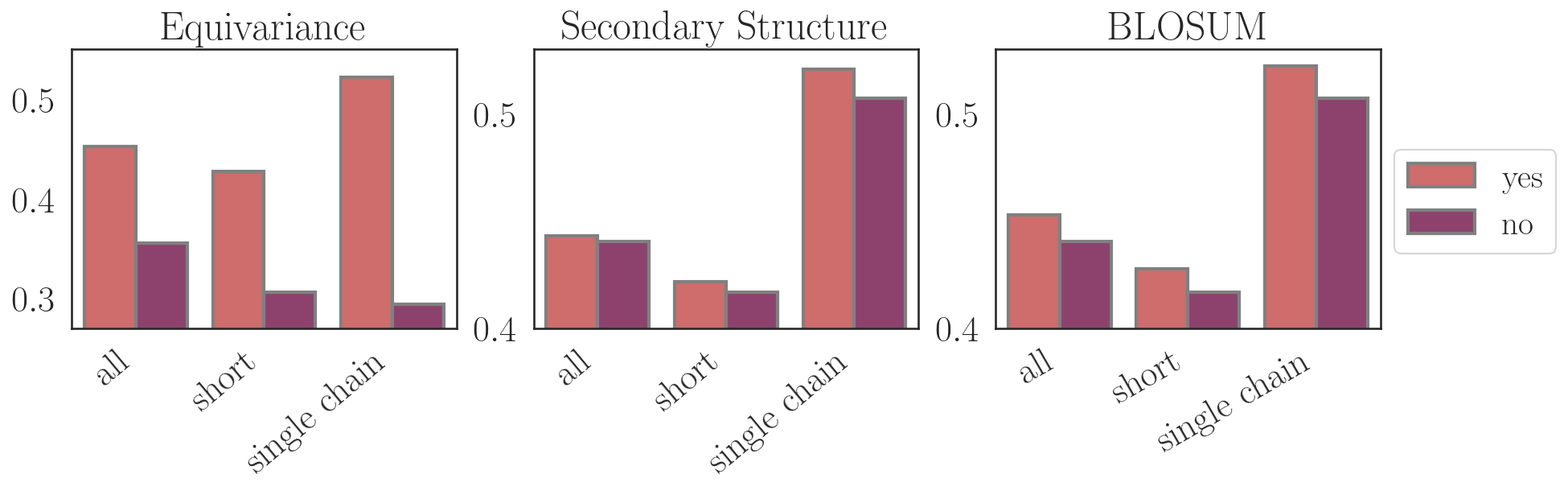}
    \caption{Recovery rate with the different selection of the transition matrix, whether considering equivariance and secondary structure.}
    \label{fig:recoverRate_SubMtr}
\end{figure}

In our sampling procedure, we accelerate the original DDPM sampling algorithm, which takes every step in the reverse sampling process, by implementing the discrete DDIM as per Equation \ref{eq:ddimPosterior}. This discrete DDIM allows us to skip every $k$ steps, resulting in a speed-up of the original DDPM by a factor of $k$. We conducted an ablation study on the impact of speed and recovery rate by trying different skip steps: $1, 2, 5, 10, 20, 25, 50, \text{and } 100$. We compare the recovery rates achieved by these different steps. Our results revealed that the recovery rate performance decays as the number of skipped steps increases. The best performance is achieved when skipping a single step, resulting in a recovery rate of $52.21\%$, but at a speed of $100$ times slower than when skipping $100$ steps, which yields a recovery rate of $47.66\%$. %\bxc{this trade-off have already been discussed in the main paper. Should've been comparison on confusion matrix and BLOSUM.}

\begin{figure}[!ht]
    \centering
    \includegraphics[width=0.8\textwidth]{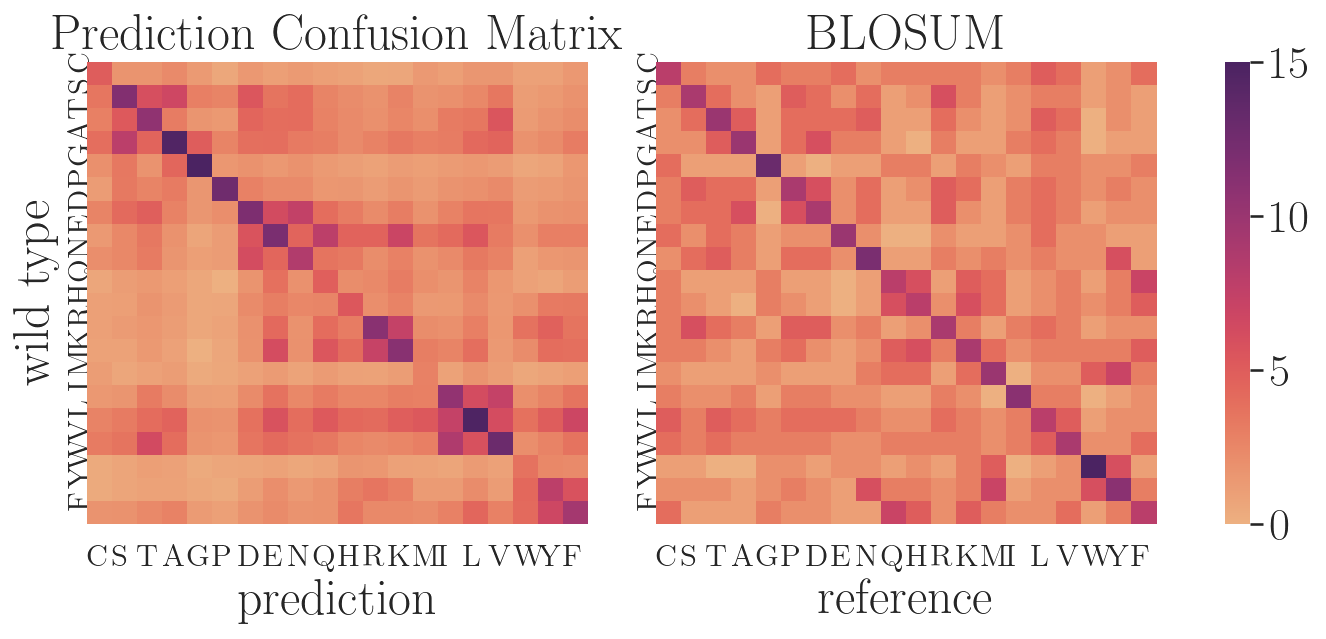}
    \caption{Comparison of the distribution of mutational prior (BLOSUM replacement matrix) and sampling results. }
\label{fig:blosum}
\end{figure}

We further compare the diversity of \ddif~with \textsc{PiFold} and \textsc{ProteinMPNN} in Figure~\ref{fig:diversity}. For a given backbone, we generate $100$ sequences with a self-similarity less than $50\%$ and employ t-SNE \cite{van2008visualizing} for projection into a $2$-dimensional space. At the same level of diversity, \ddif~ encompasses the wild-type sequence, whereas the other two methods fail to include the wild-type within their sample region. Furthermore, inspiring at a recovery rate threshold of $45\%$ for this protein, \ddif~ manages to generate a substantial number of samples, whereas the other two methods revert to deterministic results. This further substantiates the superiority of our model in terms of achieving sequence diversity and a high recovery rate concurrently.

% We also evaluated the speed-up sampling algorithm within this dataset, as depicted in Figure~\ref{fig:ddim}. As outlined in Equation \eqref{ddim}, we can bypass $k$ steps during the sampling phase. We selected a range of step sizes and assessed their performance in terms of the recovery rate and the time required to sample $1200$ sequences. The recovery rate mildly declines with the increment in step size, reaching $48.13\%$ at a step size of $100$. However, the sampling speed at a step size of $100$ is effectively $100$ times faster than at a step size of $1$, demonstrating a considerable speed-up.

\begin{figure}[H]
    \centering
    \begin{minipage}{0.9\textwidth}
        
        \includegraphics[width=\textwidth]{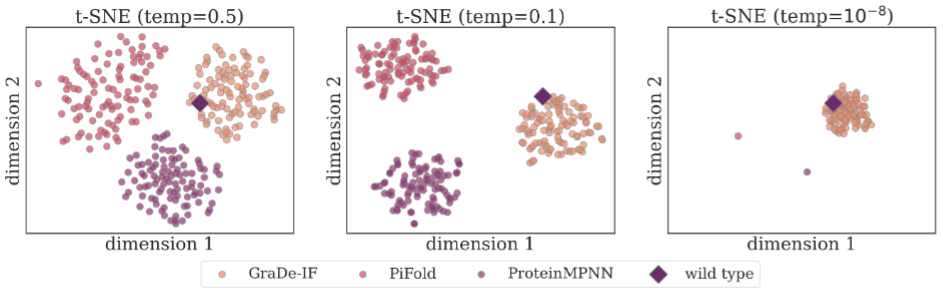}
        \caption{t-SNE of the generated sequences of \ddif~ compared to \textsc{PiFold} and \textsc{ProteinMPNN}.}
        \label{fig:diversity}
    \end{minipage}
\end{figure}

% \begin{figure}[H]
%     \centering
%         \includegraphics[width=0.4\textwidth]{figure/line_DDIM.png}
%         \caption{Trade-off of sampling speed and recovery rate.}
%         \label{fig:ddim}
% \end{figure}
\section{Additional Folding Results}
\label{sec:app:af2}
We further analyzed the generated sequences by comparing different protein folding predictions. We consider the crystal structures of three native proteins with PDB IDs: \textbf{1ud9} (A chain), \textbf{2rem} (B chain), \textbf{3drn} (B chain), which we randomly choose from \textbf{CATH} dataset. For each structure, we generated three sequences from the diffusion model and used \textsc{AlphaFold 2} \cite{jumper2021alphafold} to predict the respective structures. As shown in Figure~\ref{fig:folding+}, these predictions (in \textcolor{antiquefuchsia}{purple}) 
% \bxc{if you highlight with specific colors, you might need to give all three colors to avoid confusion here.} \ygc{I just say it in general: all predicted are in purple, either shallow or dark} 
were then compared with the structures of the native protein sequences (in \textcolor{antiquebrass}{nude}). We can observe that the RMSD for all cases is lower than the preparation accuracy of the wet experiment.
% \bxc{the resolution is case-by-case for each protein. You'll need to check it in PDB by their ID. Generally setting 2\AA as the threshold is not correct.}. 
The results demonstrate that our model-generated sequences retain the core structure, indicating their fidelity to the original structures.

\begin{figure}[!ht]
    \centering
    \includegraphics[width=\textwidth]{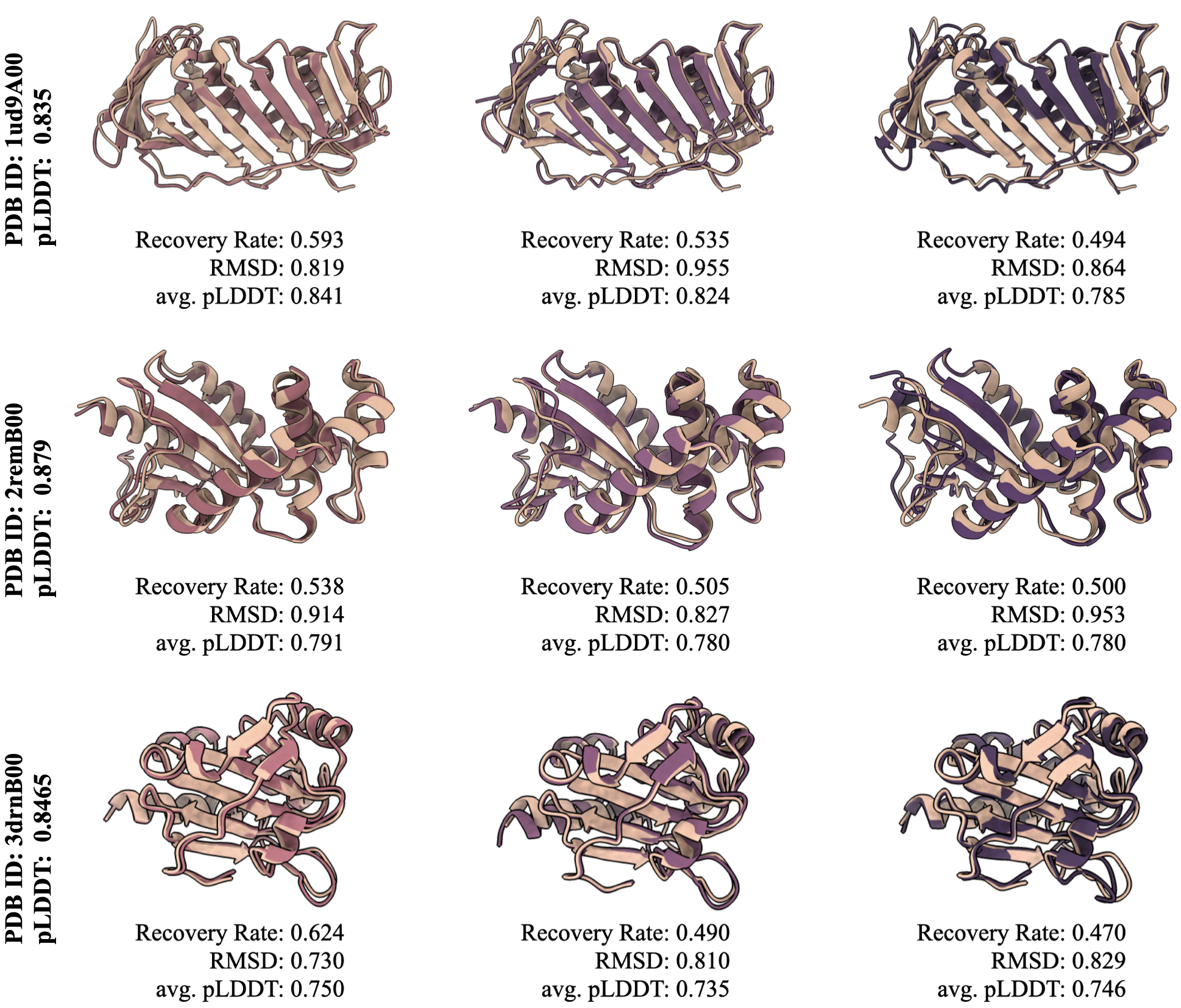}
    \caption{Folding comparisons between native sequence and generated sequence}
    \label{fig:folding+}
\end{figure}

\end{document}